%% file: main.tex
\documentclass[sigconf,nonacm]{acmart}

\usepackage{pvldb}

\newcommand{\papertitle}{Logos: Certified Order-Sensitive SQL Rewrites with Mechanized Semantics and LLM Guidance}

\renewcommand\vldbdoi{XX.XX/XXX.XX}
\renewcommand\vldbpages{XXX-XXX}
\renewcommand\vldbavailabilityurl{}

\usepackage{amsmath}
\usepackage{booktabs}
\usepackage{listings}
\usepackage{tikz}
\usepackage{xcolor}
\usepackage{colortbl}
\usepackage{xspace}
\lstdefinestyle{sql}{
  basicstyle=\footnotesize\ttfamily,
  keywordstyle=\bfseries,
  commentstyle=\itshape\color{black!55},
  columns=fullflexible,
  keepspaces=true,
  breaklines=true,
  showstringspaces=false,
  frame=none,
  aboveskip=2pt,
  belowskip=1pt,
  language=SQL
}

\lstdefinestyle{sqlpanel}{
  style=sql,
  basicstyle=\scriptsize\ttfamily,
  backgroundcolor=\color{black!7},
  frame=single,
  framerule=0pt,
  rulecolor=\color{black!7},
  framesep=0pt,
  framexleftmargin=5pt,
  framexrightmargin=5pt,
  framextopmargin=4pt,
  framexbottommargin=4pt,
  aboveskip=0pt,
  belowskip=0pt,
  breaklines=false,
  escapeinside={(*@}{@*)}
}

\newcommand{\logos}{Logos\xspace}
\newcommand{\bestbaselinerate}{64.0\%\xspace}
\newcommand{\logosrate}{86.9\%\xspace}
\newcommand{\logosimprovement}{22.9\xspace}
\newcommand{\querypaneltitle}[1]{%
  \footnotesize\bfseries\ttfamily
  \mbox{-\kern0pt-\kern0pt-\kern0pt-\kern0pt-\hspace{0.6em}%
    #1 Query\hspace{0.6em}-\kern0pt-\kern0pt-\kern0pt-\kern0pt-}}

\newcommand{\Tval}{\mathsf{T}}
\newcommand{\Fval}{\mathsf{F}}
\newcommand{\Uval}{\mathsf{U}}
\newcommand{\Eval}{\mathcal{E}}
\newcommand{\rowsbag}{\mathsf{bag}}
\definecolor{benchmarkblue}{HTML}{2F5D9B}
\definecolor{rqframe}{HTML}{8A6428}
\definecolor{rqfill}{HTML}{F4ECD9}
\newcommand{\benchheat}[2]{\cellcolor{benchmarkblue!#1}\strut #2}
\newcommand{\resultHeader}[1]{%
  \leavevmode\raisebox{-0.92\baselineskip}[0pt][0pt]{#1}}
\newcommand{\resultCount}[2]{#1\,(#2\%)}
\newtheorem{theorem}{Theorem}
\theoremstyle{definition}
\newtheorem{definition}[theorem]{Definition}
\newcommand{\supportfull}{%
  \tikz[baseline=-0.55ex]{\fill (0,0) circle[radius=0.55ex];}}
\newcommand{\supportpartial}{%
  \tikz[baseline=-0.55ex]{%
    \fill (0,0) -- (90:0.55ex)
      arc[start angle=90,end angle=270,radius=0.55ex] -- cycle;
    \draw[line width=0.3pt] (0,0) circle[radius=0.55ex];}}
\newcommand{\supportnone}{%
  \tikz[baseline=-0.55ex]{%
    \draw[line width=0.3pt] (0,0) circle[radius=0.55ex];}}

\begin{document}

\title{\papertitle}

\author{Jingyu Ke}
\orcid{0009-0008-6848-3105}
\affiliation{%
  \institution{Shanghai Jiao Tong University}
  \city{Shanghai}
  \country{China}
}
\email{Windocotber@sjtu.edu.cn}

\author{Jingyang Li}
\orcid{0000-0002-3707-4623}
\affiliation{%
  \institution{Shanghai Jiao Tong University}
  \city{Shanghai}
  \country{China}
}
\email{lijjjjj@sjtu.edu.cn}

\author{Guoqiang Li}
\orcid{0000-0001-9005-7112}
\affiliation{%
  \institution{Shanghai Jiao Tong University}
  \city{Shanghai}
  \country{China}
}
\email{li.g@sjtu.edu.cn}

\begin{abstract}
SQL rewrite verification must account for duplicate rows, observable row order,
and typed value semantics. Existing verifiers have yet to combine proofs over
database instances of arbitrary finite cardinality with an ordered-list semantics for
nested, tie-sensitive top-$k$. Unbounded systems reason primarily over bags or
handle ordering through syntax-directed restrictions, whereas bounded systems
either support only restricted top-$k$ forms or impose a deterministic ordering
rather than retain all legal tie-induced outcomes. Support for typed expression
and aggregate semantics, observable runtime errors, and integrity constraints
also remains partial.

In Rocq, we mechanize a compositional logical semantics for a typed SQL
core with order-sensitive operators, capturing all possible ordered lists and
observable SQL failures in the supported fragment. To our
knowledge, this is the first mechanized SQL semantics to combine nested,
tie-sensitive top-$k$ with a closure-based lifting from bag equivalence to
ordered-list equivalence, enabling sound reuse of bag-theoretic reasoning while preserving
compositionality across order-sensitive and correlated contexts. The
formalization further provides executable semantics for PostgreSQL-oriented
scalar and aggregate evaluation and an explicit account of integrity constraints.
Building on this semantics, we present \logos, an LLM-guided Rocq verifier for
unbounded SQL rewrite equivalence. Its agent uses a verified SQL-specific lemma
library to construct query-specific Rocq proofs. Our evaluation covers 389
query pairs from Apache Calcite optimizer tests, TPC-H and TPC-DS rewrites, and
WeTune's real-application workloads. \logos solves \logosrate of them,
compared with \bestbaselinerate for SQLSolver, the strongest baseline.
\end{abstract}

\maketitle

\vldbtopmatter

\input{sections/introduction}
\input{sections/overview}
\input{sections/approach}
\input{sections/implementation}
\input{sections/evaluation}
\input{sections/related-work}
\input{sections/conclusion}

\bibliographystyle{ACM-Reference-Format}
\bibliography{references}

\end{document}

%% file: sections/introduction.tex
\section{Introduction}
\label{sec:introduction}

Query rewriting is a central technique for improving database
performance~\cite{pirahesh1992starburst}. Reported application cases include a
$100\times$ speedup for an individual query and a $3.7\times$ reduction in
production workload latency~\cite{bai2023querybooster,sun2025rbot}. Recent
rewriting systems have also uncovered hundreds of optimization opportunities
missed by existing DBMSs in widely used open-source
applications~\cite{wang2022wetune}. Deploying such rewrites requires
establishing semantic equivalence between the original and rewritten queries.
Specifically, the two queries must have the same observable behavior on every
schema-conforming finite database instance, without any preset bound on
relation cardinality. Testing can refute this obligation by finding a
counterexample, but it cannot establish it. This gap is increasingly
consequential as optimizers combine hand-written rules, workload-derived
transformations, and LLM-generated rewrites; recent industrial evidence shows
that realistic query pairs remain beyond the coverage of current equivalence
checkers~\cite{narasayya2026qoverify}.

The central semantic challenge arises from the interaction between bag
semantics and order-sensitive operators. Some relational arguments are
permutation-invariant and admit concise multiplicity proofs. In contrast,
\texttt{ORDER BY} establishes order, whereas slicing and window operators can
consume row positions. Globally quotienting intermediate results by
permutation therefore erases behavior observable to nested,
tie-sensitive top-$k$, whereas assigning each operator one fixed output list
invents an execution order that SQL leaves unspecified. Retaining all legal
ordered lists preserves this behavior, but introduces a substantial proof
burden: even routine query-specific equivalence arguments require nontrivial
structural induction over lists, often coupled with explicit permutation
reasoning~\cite{chu2017hottsql,chu2018axiomatic}.
Sound verification also requires concrete semantics for typed scalar and
aggregate evaluation, including observable failures, together with a formal
account of integrity constraints. Modeling scalar and aggregate operations as
uninterpreted functions overapproximates their behavior and can admit spurious
counterexamples that cannot arise under concrete SQL semantics.

Existing unbounded verifiers obtain much of their power from bag semantics and
algebraic normalization~\cite{chu2017hottsql,chu2018axiomatic,ding2023sqlsolver,
wang2024qed}. Support for ordered fragments is typically syntax-directed,
requiring syntactic alignment of corresponding order-sensitive operators.
Conversely, list encodings can cover richer SQL but have so far provided
bounded guarantees for restricted order-sensitive
fragments~\cite{he2024verieql}. The missing connection is not another global
choice between lists and bags. A verifier needs a semantics that retains every
possible result order, together with a checked
criterion for using bag theory at a particular subquery.

We present \logos, a proof-producing verifier that realizes this criterion. Its
semantic foundation is FormalSQL, our order-sensitive extension of
SQLCoq~\cite{benzaken2019coqsql}, mechanized in Rocq. The extension provides a
compositional logical semantics for a typed SQL core with order-sensitive
operators, together with executable, PostgreSQL-oriented value and error
interpretations for the supported fragment. Rather than selecting one
implementation-dependent result, FormalSQL assigns each query a relation
containing all possible ordered row lists and observable SQL failures in this
fragment. Applying the order-forgetting map to every successful result yields the query's bag
projection, which preserves multiplicities while discarding order. The
projection may contain multiple distinct bags, for example when top-$k$ cuts
through a tie. We identify a semantic closure condition under which two sets of
legal ordered lists are equivalent exactly when their bag projections are
equal. At such a boundary, standard multiplicity arguments can be used locally
and then lifted through the surrounding order-sensitive context.

\logos then performs LLM-guided counterexample search. If the model predicts
non-equivalence, it must propose a concrete, schema-conforming database
instance. A pair of differing executions on that instance does not in general
establish non-equivalence. When SQL does not determine a unique order, an
engine produces only one of several legal ordered lists; the observed choice
may depend on optimizer-selected scan and join plans and physical storage
order, while the relative order of rows tied on every ordering expression
remains implementation-dependent.
Row-limiting operators can turn this freedom into different selected
subsets~\cite{postgresql18select}. For a proof-producing non-equivalence result,
\logos therefore reifies the candidate as a fixed, typed FormalSQL database
instance and requires the proof agent to establish a countermodel against the
complete outcome relations. The proof exhibits a legal outcome on one side that
no legal outcome on the other can match.

Likewise, when the model predicts equivalence, \logos requires the proof agent
to establish the fixed Rocq equivalence theorem over all schema-conforming
database instances. Building on FormalSQL, we provide a library of verified
lemmas covering general structural properties and
common operator interactions. The library keeps the agent from reproving these
reusable facts and concentrates its effort on the semantic interactions and
feature combinations specific to each query pair.

We evaluate \logos on 389 query pairs spanning Calcite, the TPC
decision-support suites, and WeTune.
\logos solves \logosrate of the pairs, improving over the strongest
baseline, SQLSolver (\bestbaselinerate), by \logosimprovement percentage
points. On the order-sensitive
subset, it solves 73.1\%, compared with 41.3\%
for the strongest baseline on that subset. These results show that
explicit treatment of order and lemma-guided proof synthesis extend
proof-producing verification to substantially broader SQL feature
combinations.

This paper makes three contributions:
\begin{itemize}
  \item We extend FormalSQL with a mechanized, compositional semantics for a
  typed SQL core with order-sensitive operators. Each query denotes all
  possible ordered lists and observable SQL failures in the supported
  fragment. The formalization provides
  executable Rocq semantics for PostgreSQL-oriented scalar and aggregate
  evaluation. This includes \texttt{NULL}-sensitive evaluation and PostgreSQL
  numeric and temporal types (e.g., \texttt{DECIMAL} and
  \texttt{TIMESTAMPTZ}), with all supported failures explicit. To our knowledge,
  this is the first mechanized SQL semantics to combine nested,
  tie-sensitive top-$k$ with a closure-based lifting from bag equivalence to
  ordered-list equivalence while preserving compositionality across order-sensitive and
  correlated contexts.

  \item We implement \logos as a proof-producing prototype for SQL equivalence
  verification. It lowers query pairs in its supported PostgreSQL fragment into
  fixed Rocq verification goals and combines LLM-guided synthesis of
  equivalence proofs and certified countermodels with a reusable SQL lemma
  library and domain-specific Rocq automation. Every resulting certificate is
  checked in isolation by the Rocq kernel.

  \item We evaluate \logos against state-of-the-art SQL equivalence verifiers
  on a diverse suite of optimizer, decision-support, and real-application
  benchmarks, including Calcite, the TPC decision-support suites, and WeTune. \logos
  solves \logosrate of the query pairs, exceeding the best-performing
  baseline, SQLSolver (\bestbaselinerate), by \logosimprovement percentage
  points. These
  results show that mechanized SQL semantics
  combined with LLM-guided proof search substantially broadens the practical
  coverage of proof-producing equivalence verification.
\end{itemize}

%% file: sections/overview.tex
\section{Overview}
\label{sec:overview}

\begin{figure}[!t]
  \centering
\begin{lstlisting}[style=sqlpanel]
(*@\hfill\querypaneltitle{Source}\hfill\null@*)
SELECT *
FROM users
WHERE active IS TRUE
  AND users.id IN (
    SELECT user_id
    FROM group_users
    WHERE group_id = 8861)
LIMIT 20;
\end{lstlisting}
\vspace{2pt}
\begin{lstlisting}[style=sqlpanel]
(*@\hfill\querypaneltitle{Target}\hfill\null@*)
SELECT users.*
FROM users
INNER JOIN group_users
  ON group_users.user_id = users.id
WHERE active IS TRUE
  AND group_users.group_id = 8861
LIMIT 20;
\end{lstlisting}
  \caption{A real-world Discourse rewrite from the WeTune benchmark that
  replaces an \texttt{IN} membership test with a join.}
  \Description{Both queries return at most twenty active users belonging to
  group 8861. The source tests membership with an IN subquery. The target
  joins users with the group-users table. Neither query specifies an order.}
  \label{fig:running-example}
\end{figure}

Figure~\ref{fig:running-example} shows a real-world Discourse rewrite from the
WeTune benchmark~\cite{wang2022wetune}.  The source selects active users whose
identifiers occur among the members of group 8861.  The target replaces this
membership test with an inner join.  The schema declares
$(\mathit{group\_id},\mathit{user\_id})$ unique in \texttt{group\_users}, so a
fixed group contributes at most one matching membership row per user.  Both
queries then retain at most 20 rows without specifying an order.

Consider a database in which the relevant rows are summarized below.  The
symbols $u_1,\ldots,u_{21}$ stand for distinct complete rows of
\texttt{users}.
\begin{center}
  \small
  \begin{tabular}{@{}lll@{}}
    \toprule
    User rows & \texttt{active} & Membership in group 8861\\
    \midrule
    $u_1,\ldots,u_{21}$ & \texttt{TRUE}  & one row per user\\
    $v$                 & \texttt{FALSE} & one row\\
    $w$                 & \texttt{TRUE}  & none\\
    \bottomrule
  \end{tabular}
\end{center}
Before applying \textsc{limit}, the source core retains exactly
$u_1,\ldots,u_{21}$: the inactive row $v$ fails the first predicate, and $w$
fails the membership test.  The target core reaches the same rows through the
join.  Uniqueness is essential here; it ensures that the join emits each $u_i$
once rather than multiplying it by several matching membership rows.  Thus,
immediately before \textsc{limit}, the two query cores produce the same bag of
21 user rows with the same multiplicities and modeled failures.

At this boundary, neither core establishes an order.  Every permutation of
the 21 rows is therefore a legal result, but all of these ordered lists have
the same underlying bag.  The collection of possible bags consequently has
one member before \textsc{limit}.

The limit makes the distinction observable.  Because any $u_i$ may occur
last, taking the first 20 rows can omit any one of the 21 users.  The complete
query therefore has 21 possible result bags, each containing a different
20-user subset.  \logos records these alternatives separately instead of
selecting one engine-dependent order or combining the alternatives into a
single bag.

Preserving this set of alternatives is necessary for equivalence reasoning,
especially when the query is embedded as a subquery.  Replacing the 21 possible
bags by one arbitrarily selected bag would make the comparison depend on that
choice: two equivalent queries could be assigned different 20-row bags even
though they admit the same complete set of alternatives.  The resulting loss
of correspondence could cause equivalent queries to be reported as
non-equivalent, and the discrepancy could propagate through the enclosing query.

\logos makes these alternatives explicit.  Write $B_{21}$ for the bag that
contains each of $u_1,\ldots,u_{21}$ once.  On the database above, its compact
summary of the possible bags at the two relevant boundaries is
\[
  \text{before \textsc{limit}: }\{B_{21}\},
  \qquad
  \text{after \textsc{limit}: }
  \{B_{21}\setminus\{u_i\}\mid 1\leq i\leq 21\},
\]
where subtraction removes the single occurrence of $u_i$.  The first set says
that all legal core results share one bag; the second keeps the 21 alternatives
created when \textsc{limit} observes their positions.  This bag summary is not
itself the query semantics.  \logos retains ordered row sequences as its
successful outcomes and keeps modeled SQL errors as separate outcomes.  It
forgets order only at a local proof boundary, without merging distinct possible
bags.

To verify the rewrite, \logos first proves that the two cores produce the same
$B_{21}$.  It then establishes that both cores are \emph{bag closed}: every
ordering of $B_{21}$ is legal.  These facts
allow the local bag proof to recover equality of all ordered core outcomes, as
formalized in Section~\ref{sec:bag-abstraction}.  Semantic congruence then
carries this equality through the common \textsc{limit} and, when present,
through any enclosing query.  Consequently, the complete queries preserve
precisely the same 21 possible bags above, while genuine errors remain visible
throughout the query context.

%% file: sections/approach.tex
\section{Formal Semantics and Semantic Decomposition}
\label{sec:technique}

The formal development of \logos is grounded in an ordered-list semantics of query
execution, from which bag reasoning is derived as a controlled abstraction.
Section~3.1 fixes
the notation used throughout the paper and presents the typed core syntax of
our supported SQL fragment.  Section~3.2 gives this syntax an error-aware
execution semantics and defines the resulting verification objective.
Section~3.3 then establishes the closure conditions under which bag equivalence
can be lifted to ordered outcome equivalence and composed with surrounding
operators.  Finally, Section~3.4 instantiates the development on the
running example from Section~\ref{sec:overview}.

\subsection{Notation and Typed Core Syntax}
\label{sec:surface-fragment}
\label{sec:core-language}

Let $\tau$ range over the supported SQL types, $\mathcal V_\tau$ over their
non-NULL values, and
$\mathcal V_\tau^\bot=\mathcal V_\tau\uplus\{\mathsf{NULL}_\tau\}$ over their
nullable values.  A typed attribute $a:\tau$ associates the column identifier
$a$ with its SQL type $\tau$.  Write $v_1\approx_\tau v_2$ when two nullable
values agree in the typed SQL value model; in particular,
$\mathsf{NULL}_\tau\approx_\tau\mathsf{NULL}_\tau$.  An ordered relation signature
$\sigma=(a_1:\tau_1,\ldots,a_m:\tau_m)$ is a finite sequence of typed
attributes with distinct identifiers and fixes both column order and types,
with
\begin{equation}
  \mathcal V_\sigma
  =\mathcal V_{\tau_1}^\bot\times\cdots\times
   \mathcal V_{\tau_m}^\bot.
  \label{eq:row-domain}
\end{equation}
Thus, $\mathcal V_\sigma$ is the domain of well-typed row values conforming to
$\sigma$, and $r\in\mathcal V_\sigma$ denotes a particular row value.  The
$i$th component of $r$ is a nullable value of type $\tau_i$; for
$a_i:\tau_i\in\sigma$, write $r[a_i]$ for this component.  Rows satisfy
$r_1\approx_\sigma r_2$ when their corresponding components are related by
the appropriate $\approx_\tau$.

\begin{definition}[Finite bag]
\label{def:finite-bag}
For a set $X$, a finite bag over $X$ is a function $B:X\to\mathbb N$ whose
support $\{x\in X\mid B(x)>0\}$ is finite; $B(x)$ is the multiplicity of $x$.
We write $\mathsf{Bag}(X)$ for the set of finite bags over $X$.  When $X$
carries a semantic equivalence $\approx_X$, let $\#_B(x)$ count all
occurrences in $B$ equivalent to $x$, and write $B_1=_{\mathsf B}B_2$ when
$\#_{B_1}(x)=\#_{B_2}(x)$ for every $x\in X$.
\end{definition}

\begin{definition}[Finite list]
\label{def:finite-list}
For a set $X$, the set of all finite lists over $X$ is
\begin{equation}
  \mathsf{List}(X)
  \triangleq\{[x_1,\ldots,x_n]\mid n\in\mathbb N\ \land\
      \forall i\in\{1,\ldots,n\}.\ x_i\in X\}.
  \label{eq:finite-list}
\end{equation}
The case $n=0$ is the empty list $[\,]$; equivalently, finite lists are
generated by $\ell::=[\,]\mid x::\ell$ for $x\in X$.  Positions and repeated
occurrences are retained.  A list $\ell=[x_1,\ldots,x_n]$ induces
$\rowsbag(\ell)\in\mathsf{Bag}(X)$ by counting occurrences.  Thus $\rowsbag$
forgets order but retains multiplicity; we write
$\{\!\{x_1,\ldots,x_n\}\!\}$ for the resulting bag.  For
$B\in\mathsf{Bag}(X)$, write $\mathsf{Rows}(B,\ell)$ when
$\rowsbag(\ell)=_{\mathsf B}B$.
\end{definition}

\begin{definition}[Schema and schema contract]
\label{def:schema-contract}
A schema $\Sigma$ maps each relation name $R$ to an ordered relation signature
$\Sigma(R)$.  A schema contract $\Gamma$ is a finite set of integrity
constraints over the relations and attributes declared by $\Sigma$.
\end{definition}

\begin{definition}[Finite database]
\label{def:finite-database}
A finite database $D$ over $\Sigma$ maps each declared relation $R$ to
$D(R)\in\mathsf{Bag}(\mathcal V_{\Sigma(R)})$.  We write
$D\models(\Sigma,\Gamma)$ when $D$ satisfies every constraint in $\Gamma$.
\end{definition}

\input{figures/core-syntax}

Figure~\ref{fig:sql-admission-core} presents the typed core syntax used by
\logos.  It covers the supported SQL query expressions, scalar expressions,
predicates, and the integrity constraints admitted by schema contracts.  The
metavariables $f$ and $\theta$ range over the supported typed scalar and
predicate operators, $n\in\mathbb N$, and $\varphi_{\rm row}$ is a well-typed,
subquery-free predicate in the scope of its declared relation.  The list
$W=[w_1,\ldots,w_u]$ contains typed window outputs
$w=(a:\tau,\psi)$, where $\psi$ is row number, a cumulative aggregate over the
current partition prefix, or an aggregate over the complete partition.
The corresponding semantics supports common aggregate functions and typed
scalar operators over the modeled SQL types.

The query-expression category $Q$ begins with two constant relations.
$\mathsf{Empty}(\sigma)$ denotes the empty relation of signature $\sigma$,
whereas $\mathsf{Singleton}$ denotes the relation of empty signature $()$ whose
sole row is the empty tuple $\langle\rangle$.  The latter supplies a unit input
from which extended projection can construct a concrete row.  Since
$\mathsf{Project}$ evaluates its list once for each input row,
$\mathsf{Singleton}$ makes that evaluation occur exactly once.  For example,
\[
  \mathsf{Project}\bigl(
    [(1\!:\!\texttt{INTEGER},x\!:\!\texttt{INTEGER})],
    \mathsf{Singleton}\bigr)
\]
has signature $(x:\texttt{INTEGER})$ and contains the single row
$\langle 1\rangle$; it is the core representation of the single-row table
constructor \texttt{VALUES (1)}.  Bag union of such terms represents a
multirow \texttt{VALUES} constructor.

The remaining query constructors cover the relational and order-sensitive
operators listed in Figure~\ref{fig:sql-admission-core}.  Section~3.2 assigns
them ordered-list semantics and specifies how modeled errors propagate
from demanded subterms.  The core set operators are the \textsc{all} variants
with bag multiplicities; surface \textsc{distinct} variants are normalized
into these operators and $\mathsf{Distinct}$.  Within this syntax, inner joins
are represented by filtering a cross join, because the separate
$\mathsf{Join}$ constructor is reserved for outer, semi, and anti joins.

\subsection{Error-Aware Execution Semantics and Equivalence}
\label{sec:observations}
\label{sec:problem}

Query execution depends on more than the database contents.  A subquery may
refer to row and group bindings supplied by its surrounding context, while
unordered inputs and unresolved sort ties may admit multiple ordered lists.
We first formalize the SQL program state that makes these dependencies
explicit.

\begin{definition}[SQL program state]
\label{def:sql-program-state}
An SQL binding is a triple
$\beta=(\sigma_\beta,m_\beta,\ell_\beta)$, where $\sigma_\beta$ is its row
signature, $m_\beta$ is either $\mathsf{Row}$ or $\mathsf{Group}(G)$, and
$\ell_\beta\in\mathsf{List}(\mathcal V_{\sigma_\beta})$.  A row binding induced
by $r\in\mathcal V_\sigma$ is
$\mathsf{row}_\sigma(r)=(\sigma,\mathsf{Row},[r])$; a group binding stores in
$\ell_\beta$ all rows consumed by aggregate expressions.  For
$g\in\mathsf{List}(\mathcal V_\sigma)$, write
$\mathsf{gbind}_{\sigma,G}(g)=(\sigma,\mathsf{Group}(G),g)$.

An SQL program state is a pair
\begin{equation}
  \Omega=(D,\rho),
  \label{eq:sql-program-state}
\end{equation}
where $D$ is the finite database and
$\rho=\beta_0::\cdots::\beta_k::\epsilon$ is a list of visible row and group
bindings, with the head $\beta_0$ denoting the innermost scope.  The top-level
environment is $\epsilon$, and an extension prepends its new binding to $\rho$.
\end{definition}

\input{figures/core-semantics}

\paragraph{Binding resolution.}
For a well-scoped typed attribute $a:\tau$, define
$i=\min\{j\mid a:\tau\in\sigma_{\beta_j}\ \land\
\ell_{\beta_j}\ne[\,]\}$.  If $\ell_{\beta_i}=r::\ell'$, then
$\mathsf{resolve}_{\rho}(a)=r[a]$.
Thus, resolution searches from the list head (the stack top) toward the outer
bindings and returns the value from the first row of the nearest nonempty
binding in which the attribute is visible.  Attribute evaluation is given by
\[
  \frac{\mathsf{resolve}_{\rho}(a)=v}
       {(D,\rho)\vdash a:\tau\Downarrow\mathsf{ok}(v)}
  \quad(\textsc{Attribute})
\]
A row binding supplies values to scalar expressions and predicates; a group
binding additionally supplies the rows consumed by aggregate expressions.
Subqueries inherit the current binding list, so an attribute not resolved by the
innermost binding may refer to an enclosing one, as required for correlated
subqueries.

Let $\mathsf B_3=\{\Tval,\Fval,\Uval\}$ and let $\mathsf{Err}$ be the domain
of modeled SQL failures.  Lowercase $q$ ranges over well-typed terms generated
by $Q$, while $V$ and $P$ range over the corresponding scalar and predicate
terms.  We write
$\Omega\vdash q\Downarrow\mathsf{ok}(\ell)$ when executing query $q$ in
program state $\Omega$ may return the ordered row list
$\ell\in\mathsf{List}(\mathcal V_\sigma)$, and write
$\Omega\vdash q\Downarrow\mathsf{err}(e)$ when it may instead fail with
$e\in\mathsf{Err}$.  The symbol $\Downarrow$ is relational: the same query may
derive multiple ordered-list outcomes.  Analogously,
$\Omega\vdash V\Downarrow\mathsf{ok}(v)$ returns
$v\in\mathcal V_\tau^\bot$, and
$\Omega\vdash P\Downarrow\mathsf{ok}(b)$ returns $b\in\mathsf B_3$; either
judgment may instead return $\mathsf{err}(e)$.
Scalar calls use their concrete typed interpretations, predicates use SQL
three-valued logic, and $\mathsf{Case}$ executes only its selected branch.  A
scalar subquery returns its sole value, returns a typed \textsc{null} on an
empty result, and raises a cardinality error on multiple rows.  The
$\mathsf{In}$ judgment applies SQL three-valued existential comparison, while
$\mathsf{Exists}$ observes nonemptiness and evaluates only the subterms needed
to determine it; both preserve errors from demanded subterms.

Let $Q_\sigma$ denote the well-typed query terms with output signature
$\sigma$, and write $\mathcal O_\sigma$ for the semantic domain of their
results:
\[
  \mathcal O_\sigma
  =\{\mathsf{ok}(\ell)\mid
      \ell\in\mathsf{List}(\mathcal V_\sigma)\}
   \uplus
   \{\mathsf{err}(e)\mid e\in\mathsf{Err}\}.
\]
For a fixed database $D$ and binding list $\rho$, query semantics is the
set-valued map
$\Eval_{D,\rho}:Q_\sigma\to\mathcal P(\mathcal O_\sigma)$, defined by
\begin{equation}
  \Eval_{D,\rho}(q)=\{o\in\mathcal O_\sigma\mid
    (D,\rho)\vdash q\Downarrow o\}.
  \label{eq:outcome-relation}
\end{equation}
Thus, a query denotes its complete relation of legal ordered lists and
observable failures, rather than one execution selected by an engine.

\paragraph{Auxiliary judgments.}
Figure~\ref{fig:core-query-semantics} defines the big-step semantics for most
of the well-typed SQL query expressions introduced in
Figure~\ref{fig:sql-admission-core}.\footnote{The complete typing and semantic
definitions, including the rules omitted from
Figure~\ref{fig:core-query-semantics}, are available in the FormalSQL
mechanization at
\url{https://github.com/WindOctober/FormalSQL}.}
We next explain the auxiliary judgments and relations used in these rules.
The superscript ${\mathsf B}$ marks an auxiliary relation that treats successful
query inputs as bags: the enclosing rule replaces each child list $\ell$ by
$\rowsbag(\ell)$, thereby discarding positions while retaining multiplicities.
The auxiliary relation operates on these bags and produces a result bag $B$,
which represents every ordered list $\ell'$ satisfying
$\mathsf{Rows}(B,\ell')$; errors are propagated unchanged.  This convention
keeps the displayed judgments in the common result domain
$\mathcal O_\sigma$.

For $\Omega=(D,\rho)$ and $r\in\mathcal V_\sigma$, define
\[
  \Omega[r]=(D,\mathsf{row}_\sigma(r)::\rho),
\]
which makes $r$ the innermost visible row.  This extension is used whenever an
auxiliary judgment evaluates an expression against one input row.

\paragraph{Bag-level auxiliary relations.}
Recall that $\mathsf{Rows}(B,\ell)$ makes $\ell$ an arbitrary ordered
representative of the bag $B$.  The operation
$\mathsf{set}^{\mathsf B}_{\mathit{op}}$ applies the
SQL bag operation designated by $\mathit{op}$; for example, the
\textsc{union all} case adds multiplicities:
\[
  \#_{\mathsf{set}^{\mathsf B}_{\mathsf{Union}}(B_1,B_2)}(r)
  =\#_{B_1}(r)+\#_{B_2}(r).
\]
The operator $\mathsf{distinct}^{\mathsf B}$ assigns multiplicity one to every
semantic row class occurring in its input bag.  For rows
$r_i\in\mathcal V_{\sigma_i}$ with disjoint attributes, let
$r_1\mathbin{\|}r_2$ be their concatenation, of signature
$\sigma_1\mathbin{\|}\sigma_2$.  The bag product satisfies
\[
  \#_{B_1\otimes B_2}(r_1\mathbin{\|}r_2)
  =\#_{B_1}(r_1)\#_{B_2}(r_2).
\]

The auxiliary join judgment
$\mathsf{join}^{\mathsf B}_{j,P}(B_1,B_2)$ first evaluates $P$ on each candidate
$r_1\mathbin{\|}r_2$ under the corresponding row binding.  A pair matches
exactly when this evaluation returns $\Tval$; $\Fval$ and $\Uval$ do not
match, and an error is propagated.  The join kind $j$ then determines the
result: left, right, and full joins add their respective null-extended
unmatched rows, while semi and anti joins retain a left row according to
whether a match exists.  The auxiliary judgment then admits every ordered
representative of the resulting bag.

Grouping illustrates how a compound operator is defined at this boundary.
Recall from Figure~\ref{fig:sql-admission-core} that
$G=[V_1,\ldots,V_s]$ is a well-typed list of scalar expressions.  Let
$\sigma_G$ be its induced ordered result signature.  Evaluating $G$ under
$\Omega[r]$ assembles a key row $k\in\mathcal V_{\sigma_G}$, written
$\Omega[r]\vdash G\Downarrow\mathsf{ok}(k)$.  The grouping relation has
signature
\[
  \mathsf{Groups}_{\Omega,G}\subseteq
  \mathsf{Bag}(\mathcal V_\sigma)\times
  \mathsf{List}(\mathsf{List}(\mathcal V_\sigma)).
\]
To derive $\mathsf{Groups}_{\Omega,G}(B,\Pi)$, choose an ordered
representative $\ell=[r_1,\ldots,r_n]$ of $B$ and evaluate one key row for
each occurrence:
\[
  \mathsf{Rows}(B,\ell),
  \qquad
  \bigwedge_{1\le i\le n}
    \Omega[r_i]\vdash G\Downarrow\mathsf{ok}(k_i).
\]
Rows whose key rows are related by $\approx_{\sigma_G}$ belong to the same
group.  Taking the first occurrence of each key as its representative gives
\[
  \Pi=\Bigl[\,[r_t\mid k_t\approx_{\sigma_G}k_i]\ \Bigm|
    \substack{1\le i\le n,\\
      \neg\exists h<i.\ k_h\approx_{\sigma_G}k_i}\,\Bigr].
\]
Both comprehensions range over input positions in increasing order, so
duplicates are retained and first-occurrence order is preserved.  The SQL
global-group convention is
$\mathsf{Groups}_{\Omega,[\,]}(\varnothing,[[\,]])$.

For $\Omega=(D,\rho)$ and $g\in\mathsf{List}(\mathcal V_\sigma)$, write
\[
  \Omega\langle G,g\rangle
  =(D,\mathsf{gbind}_{\sigma,G}(g)::\rho).
\]
The judgment $\Omega\vdash\mathsf{groups}_{L,G,P}(\Pi)\Downarrow o$
processes each $g\in\Pi$ under $\Omega\langle G,g\rangle$.  It
first finalizes the aggregate applications required by $L$ and $P$, then
evaluates $P$ as \textsc{having}.  On success it produces
\[
  \ell_g=\bigl[\,r_g\mid g\leftarrow\Pi,\;
    \Omega\langle G,g\rangle\vdash P\Downarrow\mathsf{ok}(\Tval),\;
    \Omega\langle G,g\rangle\vdash L\Downarrow\mathsf{ok}(r_g)\,\bigr].
\]
Thus $\Fval$ and $\Uval$ discard a group, while demanded failures are
propagated.

\paragraph{Order-aware list operations.}
The order-sensitive auxiliary operations act on concrete list positions.
Recall from Figure~\ref{fig:sql-admission-core} that
$L=[(V_1,a_1:\tau_1),\ldots,(V_k,a_k:\tau_k)]$.  Its left-to-right
evaluation is specified by the following rule schemes:
\[
  \frac{\strut}
       {\Omega[r]\vdash[\,]\Downarrow\mathsf{ok}(\langle\rangle)}
  \ (\textsc{List-Nil})
\]
\[
  \frac{\Omega[r]\vdash V_j\Downarrow\mathsf{ok}(v_j)
          \quad(1\le j\le k)}
       {\Omega[r]\vdash L\Downarrow
          \mathsf{ok}(\langle a_1=v_1,\ldots,a_k=v_k\rangle)}
  \ (\textsc{List-Ok})
\]
\[
  \frac{\substack{\Omega[r]\vdash V_j\Downarrow\mathsf{ok}(v_j)\\
                   1\le j<i}
        \qquad
        \Omega[r]\vdash V_i\Downarrow\mathsf{err}(e)}
       {\Omega[r]\vdash L\Downarrow\mathsf{err}(e)}
  \ (\textsc{List-Err})
\]
For $r\in\mathcal V_\sigma$ and $o\in\mathcal O_\sigma$, projection and
filtering use the following combinators:
\[
  \begin{aligned}
  \mathsf{cons}(r,o)&=
    \begin{cases}
      \mathsf{ok}(r::\ell),&o=\mathsf{ok}(\ell),\\[-1pt]
      \mathsf{err}(e),&o=\mathsf{err}(e),
    \end{cases}\\[-1pt]
  \mathsf{keep}(b,r,o)&=
    \begin{cases}
      \mathsf{cons}(r,o),&b=\Tval,\\[-1pt]
      o,&b\in\{\Fval,\Uval\}.
    \end{cases}
  \end{aligned}
\]
Projection consequently emits one row for each input position, whereas
filtering returns a stable subsequence.  Both judgments preserve the first
error encountered in a demanded row.

A sort specification $K=[\kappa_1,\ldots,\kappa_s]$ is a lexicographic list
of keys $\kappa=(a:\tau,\delta,\nu)$, where $\delta$ chooses ascending or
descending comparison and $\nu$ places \textsc{null}s first or last.  Write
$r\preceq_K r'$ when the first key on which $r$ and $r'$ differ orders $r$
before $r'$ (or all keys tie).  Then
\[
  \mathsf{sorted}_K([r_1,\ldots,r_m])
  \quad\Longleftrightarrow\quad
  \forall i<m.\ r_i\preceq_K r_{i+1}.
\]
Together, $\mathsf{Rows}(\rowsbag(\ell_0),\ell)$ and
$\mathsf{sorted}_K(\ell)$ state that $\ell$ is a $K$-sorted permutation of
the input occurrences; no order is imposed between rows that tie on all keys.
For $\ell=[r_1,\ldots,r_m]$, the positional operators are
\[
  \begin{aligned}
  \mathsf{take}_n(\ell)&=[r_1,\ldots,r_{\min(n,m)}],\\[-1pt]
  \mathsf{drop}_n(\ell)&=[r_{\min(n,m)+1},\ldots,r_m].
  \end{aligned}
\]
Here an empty index interval denotes $[\,]$.

The rules in Figure~\ref{fig:core-query-semantics} expose the order
boundary directly.
Projection and filtering transform one selected child list without reordering
it.  Table scans, bag operations, and the auxiliary judgments for joins and
grouping do not preserve child positions; they directly
admit their legal ordered outputs.  Sorting admits every bag-preserving list
satisfying $K$, so ties remain nondeterministic, whereas offset and fetch
consume concrete positions.

For any outcome relation $E\subseteq\mathcal O_\sigma$, define
\[
  \mathsf{Succ}(E)=\{\ell\mid\mathsf{ok}(\ell)\in E\},
  \qquad
  \mathsf{Errs}(E)=\{e\mid\mathsf{err}(e)\in E\}.
\]

\begin{definition}[Observation and outcome equivalence]
\label{def:outcome-equivalence}
Using the typed value and row equivalences from Section~3.1, write
$\ell_1\approx_{\mathsf L}\ell_2$ when two lists have equal length and their
rows are related positionwise by $\approx_\sigma$.  Two sets of legal lists
match in both directions when
\begin{equation}
  \begin{aligned}
  O_1\equiv_{\mathsf L}O_2\quad\Leftrightarrow\quad&
  (\forall\ell_1\in O_1.\ \exists\ell_2\in O_2.\;
      \ell_1\approx_{\mathsf L}\ell_2)\ \land\\[-2pt]
  &(\forall\ell_2\in O_2.\ \exists\ell_1\in O_1.\;
      \ell_1\approx_{\mathsf L}\ell_2).
  \end{aligned}
  \label{eq:ordered-list-equivalence}
\end{equation}
For full outcome relations over one signature, define
\begin{equation}
  \begin{aligned}
  E_1\equiv_{\mathsf{out}}E_2
  \Leftrightarrow {}&
  E_1\ne\varnothing\ \land\ E_2\ne\varnothing\\[-2pt]
  &{}\land\
  \mathsf{Succ}(E_1)\equiv_{\mathsf L}\mathsf{Succ}(E_2)\\[-2pt]
  &{}\land\
  \mathsf{Errs}(E_1)=\mathsf{Errs}(E_2).
  \end{aligned}
  \label{eq:outcome-equivalence}
\end{equation}
An outcome relation is safe when it contains at least one successful result
and no failure:
\[
  \mathsf{Safe}(E)
  \Leftrightarrow
  \mathsf{Succ}(E)\ne\varnothing\ \land\
  \mathsf{Errs}(E)=\varnothing.
\]
Safe outcome equivalence is therefore
\begin{equation}
  E_1\equiv_{\mathsf{safe}}E_2
  \Leftrightarrow
  E_1\equiv_{\mathsf{out}}E_2\ \land\
  \mathsf{Safe}(E_1)\ \land\ \mathsf{Safe}(E_2).
  \label{eq:safe-outcome-equivalence}
\end{equation}
The inhabitedness conditions exclude vacuous equivalence of malformed
relations; an equivalence certificate must establish at least one legal
outcome on each side.  Successful results never match errors, and distinct
modeled failures remain distinguishable.
\end{definition}

\label{sec:typed-sql}
Schema conformance is part of the semantic premise.  Constraints in $\Gamma$
are interpreted over the same nullable values and three-valued predicates as
query execution; row predicates and expression-index terms must also evaluate
without error when their constraint requires them.  The complete satisfaction
relation is mechanized in the FormalSQL development cited above.

\begin{definition}[Query equivalence contracts]
\label{def:query-equivalence}
For query terms $q_1$ and $q_2$ with the same output signature, the default
verification contract is unconditional, error-preserving equivalence:
\begin{equation}
  (\Sigma,\Gamma)\models q_1\equiv q_2
  \Leftrightarrow
  \begin{aligned}[t]
  &\forall D\models(\Sigma,\Gamma).\;
    \Eval_{D,\epsilon}(q_1)
      \equiv_{\mathsf{out}}
    \Eval_{D,\epsilon}(q_2).
  \end{aligned}
  \label{eq:query-equivalence}
\end{equation}
Safe equivalence additionally requires both queries to have at least one
successful outcome and no error on any conforming database:
\begin{equation}
  (\Sigma,\Gamma)\models q_1\equiv_{\mathsf{safe}}q_2
  \Leftrightarrow
  \begin{aligned}[t]
  &\forall D\models(\Sigma,\Gamma).\;
    \Eval_{D,\epsilon}(q_1)
      \equiv_{\mathsf{safe}}
    \Eval_{D,\epsilon}(q_2).
  \end{aligned}
  \label{eq:safe-query-equivalence}
\end{equation}
\end{definition}

\paragraph{Conditional extension.}
A well-typed precondition $\Phi$ restricts either contract to databases that
satisfy it.  For the default error-preserving contract, this gives
\begin{equation}
  \forall D\models(\Sigma,\Gamma).\;\Phi(D)\Rightarrow
    \Eval_{D,\epsilon}(q_1)\equiv_{\mathsf{out}}
      \Eval_{D,\epsilon}(q_2).
  \label{eq:conditional-equivalence}
\end{equation}
The safe conditional variant replaces $\equiv_{\mathsf{out}}$ with
$\equiv_{\mathsf{safe}}$.  In either case, the certificate must establish that
$\Phi$ either follows from the input contract or is jointly satisfiable with
it.  Every database quantifier ranges over arbitrary finite cardinalities.
Failed proof search reports
\emph{unknown}; it does not weaken any of these contracts.

The outcome relation remains the specification at every query boundary.
The next section characterizes when its successful ordered lists may
be replaced locally by possible bags.

\subsection{Lifting Bag Equivalence to Ordered-List Equivalence}
\label{sec:bag-abstraction}

Fix an output signature $\sigma$ and let
$O\subseteq\mathsf{List}(\mathcal V_\sigma)$ be a relation of successful
ordered observations.  Its order-forgetting projection is
\begin{equation}
  \mathsf{Bags}(O)=\{B\mid\exists\ell\in O.\;
    \rowsbag(\ell)=_{\mathsf B}B\}.
  \label{eq:bag-projection}
\end{equation}

\begin{definition}[Bag closure]
\label{def:bag-closure}
$O$ is bag closed, written $\mathsf{Closed}(O)$, when
\begin{equation}
  \forall \ell_{\mathrm{req}}.\;
  \rowsbag(\ell_{\mathrm{req}})\in\mathsf{Bags}(O)
  \Longrightarrow
  \exists \ell_{\mathrm{obs}}\in O.\;
    \ell_{\mathrm{req}}\approx_{\mathsf L}\ell_{\mathrm{obs}}.
  \label{eq:bag-closed}
\end{equation}
\end{definition}

The following theorem is the central lifting result: at bag-closed boundaries,
equality of possible bags and failures characterizes outcome equivalence over
ordered lists.
\begin{theorem}[Bag-to-list lifting]
\label{thm:outcome-bag-bridge}
Let $E_1$ and $E_2$ be inhabited outcome relations over the same signature,
and let $O_i=\mathsf{Succ}(E_i)$.  If
$\mathsf{Closed}(O_1)$ and $\mathsf{Closed}(O_2)$, then
\begin{equation}
  E_1\equiv_{\mathsf{out}}E_2
  \Leftrightarrow
  \mathsf{Bags}(O_1)=\mathsf{Bags}(O_2)\ \land\
  \mathsf{Errs}(E_1)=\mathsf{Errs}(E_2).
  \label{eq:bag-to-list-lifting}
\end{equation}
\end{theorem}
\begin{proof}
First observe that position-wise list equivalence preserves row
multiplicities:
\begin{equation*}
  \ell_1\approx_{\mathsf L}\ell_2
  \Longrightarrow
  \rowsbag(\ell_1)=_{\mathsf B}\rowsbag(\ell_2).
  \tag{$\ast$}
\end{equation*}
Indeed, the lists have the same length and pair equivalent rows at every
position, so each row-equivalence class has the same multiplicity.

For the forward direction, assume $E_1\equiv_{\mathsf{out}}E_2$.
Equation~\eqref{eq:outcome-equivalence} immediately gives
\[
  \mathsf{Errs}(E_1)=\mathsf{Errs}(E_2),
  \qquad
  O_1\equiv_{\mathsf L}O_2.
\]
Let $B\in\mathsf{Bags}(O_1)$.  By
Equation~\eqref{eq:bag-projection},
\[
  \exists\ell_1\in O_1.\;
  \rowsbag(\ell_1)=_{\mathsf B}B.
\]
Forward list matching then supplies an $\ell_2\in O_2$ such that
\[
  \ell_1\approx_{\mathsf L}\ell_2
  \overset{(\ast)}{\Longrightarrow}
  \rowsbag(\ell_1)=_{\mathsf B}\rowsbag(\ell_2)
  \Longrightarrow
  B\in\mathsf{Bags}(O_2),
\]
where the last implication also uses symmetry and transitivity of
$=_{\mathsf B}$.  Thus
$\mathsf{Bags}(O_1)\subseteq\mathsf{Bags}(O_2)$.  The backward matching clause
gives the reverse inclusion, and hence
$\mathsf{Bags}(O_1)=\mathsf{Bags}(O_2)$.

Conversely, assume equality of the two bag projections and of the two error
sets.  For any $\ell_1\in O_1$, reflexivity of $=_{\mathsf B}$ and equality of
the projections give
\[
  \ell_1\in O_1
  \Longrightarrow
  \rowsbag(\ell_1)\in\mathsf{Bags}(O_1)
  =\mathsf{Bags}(O_2).
\]
Since $O_2$ is closed, Definition~\ref{def:bag-closure} now supplies
\[
  \exists\ell_2\in O_2.\;
  \ell_1\approx_{\mathsf L}\ell_2.
\]
Thus every observation in $O_1$ has a matching observation in $O_2$.  Applying
the same argument with the sides reversed and using closure of $O_1$ gives
\[
  O_1\equiv_{\mathsf L}O_2.
\]
Finally, the assumed inhabitedness of $E_1$ and $E_2$, together with
error-set equality, discharges all clauses of
Equation~\eqref{eq:outcome-equivalence}.
\end{proof}

\label{sec:closure}
This closure property arises throughout the core query language.  In fact, let
\[
  \begin{aligned}
  \mathcal R=\{&
    \mathsf{Empty},\mathsf{Singleton},\mathsf{Table},\mathsf{Set},\\[-2pt]
    &\mathsf{CrossJoin},\mathsf{Join},\mathsf{Distinct},\mathsf{Group},\\[-2pt]
    &\mathsf{GroupingSets},\mathsf{Rank},\mathsf{Window}\}.
  \end{aligned}
\]
We define the syntactic certificate as
\begin{equation}
  \mathsf{cert}(q)=
  \begin{cases}
    \mathsf{true}, & \text{the outer constructor of $q$ is in $\mathcal R$},\\
    \mathsf{false}, & \text{otherwise}.
  \end{cases}
  \label{eq:closure-certificate}
\end{equation}

\begin{theorem}[Closure-certificate soundness]
\label{thm:certificate-soundness}
For every admitted $q$,
\begin{equation}
  \mathsf{cert}(q)=\mathsf{true}
  \Longrightarrow
  \forall D,\rho.\;
  \mathsf{Closed}\bigl(\mathsf{Succ}(\Eval_{D,\rho}(q))\bigr).
  \label{eq:closure-certificate-soundness}
\end{equation}
\end{theorem}
\begin{proof}[Proof sketch]
Inspect the execution rule for each constructor in $\mathcal R$.  Its successful
result does not fix the order of the returned row list: once one list represents
the computed bag, every permutation of that list is also admitted.  Hence its
successful observations contain every ordering of each possible result bag,
which is precisely the closure property in
Definition~\ref{def:bag-closure}.
\end{proof}
In contrast, $\mathsf{Project}$ and $\mathsf{Filter}$ preserve the order
selected by their child rather than reset it: applied to $[r_1,r_2]$, they
retain the corresponding projected or retained rows in that order, without
independently admitting the reverse order.
Establishing closure for either operator therefore requires additional premises
about its child and row-level expressions.  A failed certificate leaves the
ordered-list semantics unchanged and does not authorize forgetting order.

Once Theorem~\ref{thm:outcome-bag-bridge} recovers outcome equivalence at
a bag-closed boundary, semantic congruence permits an outcome-equivalent child
to be replaced inside an enclosing operator.  For $\mathsf{Fetch}$, the property
used below is
\[
  \begin{aligned}
  &\Eval_{D,\rho}(q_1)\equiv_{\mathsf{out}}\Eval_{D,\rho}(q_2)\\[-2pt]
  &\quad\Longrightarrow
    \Eval_{D,\rho}(\mathsf{Fetch}(n,q_1))
    \equiv_{\mathsf{out}}
    \Eval_{D,\rho}(\mathsf{Fetch}(n,q_2)).
  \end{aligned}
\]
It follows because $\mathsf{take}_n$ preserves position-wise list equivalence
and $\mathsf{Fetch}$ propagates child errors unchanged.  FormalSQL mechanizes
the analogous congruence lemmas for query, scalar, and predicate contexts.

\subsection{Formalizing the Running Example}
\label{sec:running-example-proof}

We now instantiate the preceding results for the rewrite in
Figure~\ref{fig:running-example}.  Fix a schema $\Sigma$ containing
\textit{users} and \textit{group\_users}.  Let $L_U$ be the complete typed
projection list for \textit{users}.  The following well-typed core predicates
are obtained by lowering the corresponding SQL conditions:
\begin{itemize}
  \item $P_A \mathrel{:=} \texttt{users.active IS TRUE}$;
  \item $P_G \mathrel{:=} \texttt{group\_users.group\_id = 8861}$; and
  \item $P_J \mathrel{:=} \texttt{users.id = group\_users.user\_id}$.
\end{itemize}
The two query cores and their enclosing limits are
\begin{equation}
  \begin{aligned}
  U&=\mathsf{Filter}(P_A,\mathsf{Table}(\mathit{users})),\\
  M&=\mathsf{Filter}(P_G,\mathsf{Table}(\mathit{group\_users})),\\
  I&=\mathsf{Project}([\mathit{user\_id}],M),\\
  C_{\rm src}&=\mathsf{Filter}(\mathsf{In}((\mathit{id}),I),U),\\
  C_{\rm tgt}&=\mathsf{Project}
    (L_U,\mathsf{Filter}(P_J,\mathsf{CrossJoin}(U,M))),\\
  S&=\mathsf{Fetch}(20,C_{\rm src}),\qquad
  T=\mathsf{Fetch}(20,C_{\rm tgt}).
  \end{aligned}
  \label{eq:running-core-terms}
\end{equation}
Here \texttt{LIMIT 20} lowers to $\mathsf{Fetch}(20,\cdot)$ without an
$\mathsf{OrderBy}$, and hence consumes an unconstrained child order.  Assume
that $\Gamma$ contains the benchmark's NOT-NULL declarations, the primary key
on \textit{users.id}, and a unique key on the
\textit{group\_users} attribute pair
$(\mathit{group\_id},\mathit{user\_id})$.

Fix $D\models(\Sigma,\Gamma)$ and a binding list $\rho$, and write
$\Omega=(D,\rho)$.  For each of
$U,M,C_{\rm src}$, and $C_{\rm tgt}$, all successful observations differ only
in order; write $B_U,B_M,B_{\rm src}$, and $B_{\rm tgt}$ for their respective
underlying result bags.  For a user row $u$, define
the number of matching membership occurrences by
\begin{equation}
  N(u)=\sum_{m\in\operatorname{supp}(B_M)}
    \#_{B_M}(m)\cdot
    \mathbf 1\!\left[
      \Omega[u\mathbin{\|}m]\vdash
      P_J\Downarrow\mathsf{ok}(\Tval)\right].
  \label{eq:running-match-count}
\end{equation}
The indicator is one exactly when $P_J$ evaluates to $\Tval$ under the joined
row binding for $u$ and $m$, and zero otherwise.
The source retains each occurrence of $u$ iff $N(u)>0$, whereas the target
produces one projected occurrence for every match.  Their multiplicities are
therefore
\begin{equation}
  \#_{B_{\rm src}}(u)=\#_{B_U}(u)\mathbf 1[N(u)>0],
  \qquad
  \#_{B_{\rm tgt}}(u)=\#_{B_U}(u)N(u).
  \label{eq:running-core-multiplicities}
\end{equation}
Since $M$ restricts $\mathit{group\_id}$ to 8861, the unique key gives
$N(u)\in\{0,1\}$.  Equation~\eqref{eq:running-core-multiplicities} thus
establishes equal multiplicities for every user row.  Moreover, all scalar
operations in $P_A$, $P_G$, and $P_J$ are total on their well-typed row
bindings.  Uniformly in $D$ and $\rho$, we obtain
\begin{equation}
  \begin{aligned}
  \mathsf{Bags}(\mathsf{Succ}(\Eval_{D,\rho}(C_{\rm src})))
    &=\mathsf{Bags}(\mathsf{Succ}(\Eval_{D,\rho}(C_{\rm tgt}))),\\[-2pt]
  \mathsf{Errs}(\Eval_{D,\rho}(C_{\rm src}))
    &=\mathsf{Errs}(\Eval_{D,\rho}(C_{\rm tgt}))=\varnothing.
  \end{aligned}
  \label{eq:running-core-local-obligations}
\end{equation}

The successful observations of both cores are bag closed.  For
$C_{\rm src}$, the table scan admits every input permutation, the filters act
pointwise, and $\mathsf{In}$ depends on membership in $I$ rather than the order
of its rows.  For $C_{\rm tgt}$, $\mathsf{CrossJoin}$ resets order, after which
pointwise filtering and projection commute with permutations.  The closure
lemmas of Section~\ref{sec:closure} therefore give
\begin{equation}
  \mathsf{Closed}\bigl(\mathsf{Succ}(\Eval_{D,\rho}(C_{\rm src}))\bigr),
  \qquad
  \mathsf{Closed}\bigl(\mathsf{Succ}(\Eval_{D,\rho}(C_{\rm tgt}))\bigr).
  \label{eq:running-core-closure}
\end{equation}
The outcome relations are inhabited, so
Theorem~\ref{thm:outcome-bag-bridge}, together with
Equations~\eqref{eq:running-core-local-obligations} and
\eqref{eq:running-core-closure}, yields
\begin{equation}
  \Eval_{D,\rho}(C_{\rm src})
  \equiv_{\mathsf{out}}
  \Eval_{D,\rho}(C_{\rm tgt}).
  \label{eq:running-core-outcome-equivalence}
\end{equation}
Applying this $\mathsf{Fetch}$ property to
Equation~\eqref{eq:running-core-outcome-equivalence} gives outcome equivalence
for $S$ and $T$.  On the 21-user instance from
Section~\ref{sec:overview}, this gives the same 21 possible 20-row bags on both
sides.

\input{tables/table1-benchmark-characteristics}

\begin{theorem}[Rewrite equivalence]
\label{thm:running-example-equivalence}
Under the schema contract $(\Sigma,\Gamma)$ above,
\begin{equation}
  (\Sigma,\Gamma)\models S\equiv T.
  \label{eq:running-example-equivalence}
\end{equation}
\end{theorem}
\begin{proof}
Fix $D\models(\Sigma,\Gamma)$.  Instantiating
Equation~\eqref{eq:running-core-outcome-equivalence} at $\rho=\epsilon$ and
applying the $\mathsf{Fetch}$ property above gives
\[
  \Eval_{D,\epsilon}(S)
  \equiv_{\mathsf{out}}
  \Eval_{D,\epsilon}(T).
\]
The claim follows from Definition~\ref{def:query-equivalence}, since $D$ was
arbitrary and $S,T$ have the same output signature.
\end{proof}

%% file: figures/core-syntax.tex
\begin{figure*}[!t]
  \centering
  \begingroup
  \footnotesize
  \setlength{\arraycolsep}{2.2pt}
  \renewcommand{\arraystretch}{1.12}
  \resizebox{\textwidth}{!}{$
  \begin{array}{@{}c@{\qquad\qquad}c@{}}
    \begin{array}[t]{c}
      \text{\normalfont\scshape Query Expressions} \\[2pt]
      \begin{array}{r@{\;}c@{\;}l@{\quad}l}
      Q & ::= & \mathsf{Empty}(\sigma)
          & \textit{No rows of signature $\sigma$} \\
        & \mid & \mathsf{Singleton}
          & \textit{One empty row of signature $()$} \\
        & \mid & \mathsf{Table}(R)
          & \textit{Base relation $R$} \\
        & \mid & \mathsf{Set}(\mathit{OP}_{\mathrm{set}},Q_1,Q_2)
          & \textit{Order-insensitive bag operation} \\
        & \mid & \mathsf{CrossJoin}(Q_1,Q_2)
          & \textit{Cartesian product} \\
        & \mid & \mathsf{Join}(\mathit{OP}_{\mathrm{join}},P,Q_1,Q_2)
          & \textit{Outer/semi/anti join on $P$} \\
        & \mid & \mathsf{Project}(L,Q)
          & \textit{Order-preserving projection} \\
        & \mid & \mathsf{Filter}(P,Q)
          & \textit{Three-valued selection} \\
        & \mid & \mathsf{Distinct}(Q)
          & \textit{Duplicate elimination} \\
        & \mid & \mathsf{Group}(L,G,P,Q)
          & \textit{Grouped aggregation with \textsc{having}} \\
        & \mid & \mathsf{GroupingSets}(\mathcal G,Q)
          & \textit{Shared-input grouping-set branches} \\
        & \mid & \mathsf{Rank}(K_p,K,
          a_{\rm rank}\!:\!\texttt{BIGINT},Q)
          & \textit{\textsc{rank} window} \\
        & \mid & \mathsf{Window}(K_p,K,W,Q)
          & \textit{Shared window evaluation} \\
        & \mid & \mathsf{OrderBy}(K,Q)
          & \textit{Constrain output order} \\
        & \mid & \mathsf{Offset}(n,Q)
          & \textit{Drop first $n$ rows} \\
        & \mid & \mathsf{Fetch}(n,Q)
          & \textit{Keep first $n$ rows} \\[5pt]
      \sigma & ::= & (a_1\!:\!\tau_1,\ldots,a_m\!:\!\tau_m)
          & \textit{Ordered relation signature} \\
      L & ::= & [(V_1,a_1\!:\!\tau_1),\ldots,(V_k,a_k\!:\!\tau_k)]
          & \textit{Projection list} \\
      G & ::= & [V_1,\ldots,V_g]
          & \textit{Grouping expressions} \\
      \mathcal G & ::= & [(L_1,G_1),\ldots,(L_s,G_s)]
          & \textit{Nonempty grouping-set branch list} \\
      K_p,K & ::= & [\kappa_1,\ldots,\kappa_r]
          & \textit{Partition/order specification} \\
      \kappa & ::= & (a\!:\!\tau,\delta,\nu)
          & \textit{Typed sort key} \\
      \delta & ::= & \mathsf{Asc}\mid\mathsf{Desc}
          & \textit{Sort direction} \\
      \nu & ::= & \mathsf{First}\mid\mathsf{Last}
          & \textit{NULL placement} \\
      \mathit{OP}_{\mathrm{set}} & ::= &
          \mathsf{Union}\mid\mathsf{Intersect}\mid\mathsf{Except}
          & \textit{ALL-mode bag-operation kind} \\
      \mathit{OP}_{\mathrm{join}} & ::= &
          \mathsf{Left}\mid\mathsf{Right}\mid\mathsf{Full}
          \mid\mathsf{Semi}\mid\mathsf{Anti}
          & \textit{Join kind}
      \end{array}
    \end{array}
    &
    \begin{array}[t]{c}
      \text{\normalfont\scshape Scalar Expressions} \\[2pt]
      \begin{array}{r@{\;}c@{\;}l@{\quad}l}
      V & ::= & c:\tau\mid a:\tau
          & \textit{Typed literal or attribute} \\
        & \mid & \alpha_\omega(V):\tau
          & \textit{Aggregate application} \\
        & \mid & f(V_1,\ldots,V_r):\tau
          & \textit{Scalar application} \\
        & \mid & \mathsf{Case}(P,V_1,V_2):\tau
          & \textit{Lazy \textsc{case}} \\
        & \mid & \mathsf{Subquery}(Q):\tau
          & \textit{Scalar subquery} \\[6pt]
      \alpha & ::= & \mathsf{Count}\mid\mathsf{Sum}\mid\mathsf{Avg}
          \mid\mathsf{Min}\mid\mathsf{Max}\mid\cdots
          & \textit{Aggregate operator} \\
      \omega & ::= & \mathsf{All}\mid\mathsf{Distinct}
          & \textit{Aggregate duplicate policy} \\
      f & ::= & +\mid-\mid\times\mid\div\mid
          \mathsf{Cast}\mid\mathsf{Concat}\mid\cdots
          & \textit{Scalar operator} \\[6pt]
      \multicolumn{4}{c}{\text{\normalfont\scshape Predicates}}
          \\[2pt]
      P & ::= & \mathsf{True}
          & \textit{Truth constant} \\
        & \mid & \theta(V_1,\ldots,V_r)
          & \textit{Scalar predicate} \\
        & \mid & \mathsf{Not}(P)
          & \textit{Three-valued negation} \\
        & \mid & \mathsf{And}([P_1,\ldots,P_h])
          & \textit{SQL conjunction} \\
        & \mid & \mathsf{Or}([P_1,\ldots,P_h])
          & \textit{SQL disjunction} \\
        & \mid & \mathsf{In}((V_1,\ldots,V_r),Q)
          & \textit{Tuple membership} \\
        & \mid & \mathsf{Exists}(Q)
          & \textit{Query nonemptiness} \\[6pt]
      \multicolumn{4}{c}{\text{\normalfont\scshape Integrity Constraints}}
          \\[2pt]
      C & ::= & \mathsf{NotNull}(R,[a_1,\ldots,a_h])
          & \textit{Nullability constraint} \\
        & \mid & \mathsf{PrimaryKey}(R,[a_1,\ldots,a_h])
          & \textit{Primary key} \\
        & \mid & \mathsf{Unique}(R,[a_1,\ldots,a_h])
          & \textit{Unique key} \\
        & \mid & \mathsf{ForeignKey}\bigl(R,[a_1,\ldots,a_h],
          & \\
        &      & \qquad R',[b_1,\ldots,b_h]\bigr)
          & \textit{Foreign key} \\
        & \mid & \mathsf{Check}(R,\varphi_{\rm row})
          & \textit{Row-local check} \\
        & \mid & \mathsf{UniqueIndex}\bigl(R,[V_1,\ldots,V_u],
          \varphi_{\rm row}\bigr)
          & \textit{Expression/partial unique index}
      \end{array}
    \end{array}
  \end{array}
  $}
  \endgroup
  \caption{Typed core syntax for queries and integrity constraints.}
  \Description{Query, scalar, predicate, and constraint grammar for the typed
  core language.}
  \label{fig:sql-admission-core}
\end{figure*}

%% file: figures/core-semantics.tex
\begin{figure*}[!t]
  \centering
  \footnotesize
  \setlength{\abovedisplayskip}{5pt}
  \setlength{\belowdisplayskip}{5pt}
  \setlength{\abovedisplayshortskip}{5pt}
  \setlength{\belowdisplayshortskip}{5pt}
  \begin{minipage}[t]{0.485\textwidth}
  \centering
  \[
    \Omega\vdash\mathsf{Empty}(\sigma)\Downarrow\mathsf{ok}([\,])
    \quad(\textsc{Empty})
  \]
  \[
    \Omega\vdash\mathsf{Singleton}\Downarrow
      \mathsf{ok}([\langle\rangle])
    \quad(\textsc{Singleton})
  \]
  \[
    \frac{\mathsf{Rows}(D(R),\ell)}
         {\Omega\vdash\mathsf{Table}(R)\Downarrow\mathsf{ok}(\ell)}
    \quad(\textsc{Table})
  \]
  \[
    \frac{
      \begin{array}{c}
      \Omega\vdash q_1\Downarrow\mathsf{ok}(\ell_1)\qquad
      \Omega\vdash q_2\Downarrow\mathsf{ok}(\ell_2)\\
      \mathsf{Rows}(\mathsf{set}_{\mathit{op}}^{\mathsf B}
        (\rowsbag(\ell_1),\rowsbag(\ell_2)),\ell)
      \end{array}}
      {\Omega\vdash\mathsf{Set}(\mathit{op},q_1,q_2)
        \Downarrow\mathsf{ok}(\ell)}
    \quad(\textsc{Set})
  \]
  \end{minipage}\hfill
  \begin{minipage}[t]{0.485\textwidth}
  \centering
  \[
    \frac{
      \begin{array}{c}
      \Omega\vdash q\Downarrow\mathsf{ok}(\ell_0)\\
      \mathsf{Rows}(\mathsf{distinct}^{\mathsf B}
        (\rowsbag(\ell_0)),\ell)
      \end{array}}
      {\Omega\vdash\mathsf{Distinct}(q)\Downarrow\mathsf{ok}(\ell)}
    \quad(\textsc{Distinct})
  \]
  \[
    \frac{
      \begin{array}{c}
      \Omega\vdash q_1\Downarrow\mathsf{ok}(\ell_1)\qquad
      \Omega\vdash q_2\Downarrow\mathsf{ok}(\ell_2)\\
      \Omega\vdash\mathsf{join}^{\mathsf B}_{j,P}
        (\rowsbag(\ell_1),\rowsbag(\ell_2))\Downarrow o
      \end{array}}
      {\Omega\vdash\mathsf{Join}(j,P,q_1,q_2)\Downarrow o}
    \quad(\textsc{Join})
  \]
  \[
    \frac{
      \begin{array}{c}
      \Omega\vdash q_1\Downarrow\mathsf{ok}(\ell_1)\qquad
      \Omega\vdash q_2\Downarrow\mathsf{ok}(\ell_2)\\
      \mathsf{Rows}(\rowsbag(\ell_1)\otimes\rowsbag(\ell_2),\ell)
      \end{array}}
      {\Omega\vdash\mathsf{CrossJoin}(q_1,q_2)
        \Downarrow\mathsf{ok}(\ell)}
    \quad(\textsc{Cross})
  \]
  \end{minipage}
  \par\vspace{2pt}
  \begin{minipage}[t]{0.485\textwidth}
  \centering
  \[
    \frac{\Omega\vdash q\Downarrow\mathsf{ok}(\ell)\qquad
          \Omega\vdash\mathsf{project}_L(\ell)\Downarrow o}
         {\Omega\vdash\mathsf{Project}(L,q)\Downarrow o}
    \quad(\textsc{Project})
  \]
  \[
    \frac{\strut}
         {\Omega\vdash\mathsf{project}_L([\,])
           \Downarrow\mathsf{ok}([\,])}
    \quad(\textsc{Proj-Nil})
  \]
  \[
    \frac{\Omega[r]\vdash L\Downarrow\mathsf{err}(e)}
         {\Omega\vdash\mathsf{project}_L(r::\ell)
           \Downarrow\mathsf{err}(e)}
    \quad(\textsc{Proj-Err})
  \]
  \[
    \frac{\Omega[r]\vdash L\Downarrow\mathsf{ok}(r')\qquad
          \Omega\vdash\mathsf{project}_L(\ell)\Downarrow o}
         {\Omega\vdash\mathsf{project}_L(r::\ell)
           \Downarrow\mathsf{cons}(r',o)}
    \quad(\textsc{Proj-Cons})
  \]
  \end{minipage}\hfill
  \begin{minipage}[t]{0.485\textwidth}
  \centering
  \[
    \frac{\Omega\vdash q\Downarrow\mathsf{ok}(\ell)\qquad
          \Omega\vdash\mathsf{filter}_P(\ell)\Downarrow o}
         {\Omega\vdash\mathsf{Filter}(P,q)\Downarrow o}
    \quad(\textsc{Filter})
  \]
  \[
    \frac{\strut}
         {\Omega\vdash\mathsf{filter}_P([\,])
           \Downarrow\mathsf{ok}([\,])}
    \quad(\textsc{Filter-Nil})
  \]
  \[
    \frac{\Omega[r]\vdash P\Downarrow\mathsf{err}(e)}
         {\Omega\vdash\mathsf{filter}_P(r::\ell)
           \Downarrow\mathsf{err}(e)}
    \quad(\textsc{Filter-Err})
  \]
  \[
    \frac{\Omega[r]\vdash P\Downarrow\mathsf{ok}(b)\qquad
          \Omega\vdash\mathsf{filter}_P(\ell)\Downarrow o}
         {\Omega\vdash\mathsf{filter}_P(r::\ell)
           \Downarrow\mathsf{keep}(b,r,o)}
    \quad(\textsc{Filter-Cons})
  \]
  \end{minipage}
  \par\vspace{2pt}
  \begin{minipage}[t]{0.485\textwidth}
  \centering
  \[
    \frac{
      \begin{array}{c}
      \Omega\vdash q\Downarrow\mathsf{ok}(\ell_0)\qquad
      \mathsf{Groups}_{\Omega,G}(\rowsbag(\ell_0),\Pi)\\
      \Omega\vdash\mathsf{groups}_{L,G,P}(\Pi)
        \Downarrow\mathsf{ok}(\ell_g)\\
      \mathsf{Rows}(\rowsbag(\ell_g),\ell)
      \end{array}}
      {\Omega\vdash\mathsf{Group}(L,G,P,q)\Downarrow\mathsf{ok}(\ell)}
    \quad(\textsc{Group})
  \]
  \[
    \frac{\Omega\vdash q\Downarrow\mathsf{err}(e)}
         {\Omega\vdash U(q)\Downarrow\mathsf{err}(e)}
    \quad(\textsc{Unary-Err})
  \]
  \[
    \frac{\Omega\vdash q_1\Downarrow\mathsf{err}(e)}
         {\Omega\vdash B(q_1,q_2)\Downarrow\mathsf{err}(e)}
    \quad(\textsc{Left-Err})
  \]
  \end{minipage}\hfill
  \begin{minipage}[t]{0.485\textwidth}
  \centering
  \[
    \frac{
      \Omega\vdash q\Downarrow\mathsf{ok}(\ell_0)\qquad
      \mathsf{Rows}(\rowsbag(\ell_0),\ell)\qquad
      \mathsf{sorted}_K(\ell)}
      {\Omega\vdash\mathsf{OrderBy}(K,q)\Downarrow\mathsf{ok}(\ell)}
    \quad(\textsc{Order})
  \]
  \[
    \frac{\Omega\vdash q\Downarrow\mathsf{ok}(\ell)}
         {\Omega\vdash\mathsf{Offset}(n,q)
           \Downarrow\mathsf{ok}(\mathsf{drop}_n(\ell))}
    \quad(\textsc{Offset})
  \]
  \[
    \frac{\Omega\vdash q\Downarrow\mathsf{ok}(\ell)}
         {\Omega\vdash\mathsf{Fetch}(n,q)
           \Downarrow\mathsf{ok}(\mathsf{take}_n(\ell))}
    \quad(\textsc{Fetch})
  \]
  \vspace{-6pt}
  \[
    \frac{\Omega\vdash q_1\Downarrow\mathsf{ok}(\ell_1)\qquad
          \Omega\vdash q_2\Downarrow\mathsf{err}(e)}
         {\Omega\vdash B(q_1,q_2)\Downarrow\mathsf{err}(e)}
    \quad(\textsc{Right-Err})
  \]
  \end{minipage}
  \caption{Big-step ordered-list semantic rules for part of the core query
  language.
  Here $U$ and $B$ range over well-typed unary and left-to-right binary query
  constructors.  All metavariables are implicitly quantified over well-typed
  objects; in \textsc{Group}, $\Pi$ ranges over lists of row groups and
  $\ell_g$ over lists of rows produced from those groups.  The auxiliary
  judgments are explained in Section~\ref{sec:observations}.}
  \Description{Named inference rules define program execution under SQL
  program state Omega for constants, tables, relational operators, grouping,
  ordering, positional consumers, and propagation of child errors.}
  \label{fig:core-query-semantics}
\end{figure*}%

%% file: tables/table1-benchmark-characteristics.tex
\input{tables/table1-benchmark-statistics}
\begin{table*}[!t]
  \caption{Composition and input-query characteristics of the 389 benchmark pairs.}
  \label{tab:benchmark-characteristics}
  \centering
  \footnotesize
  \setlength{\tabcolsep}{2.6pt}
  \renewcommand{\arraystretch}{1.12}
  \begin{tabular*}{\textwidth}{@{\extracolsep{\fill}}lrcccccccccc@{}}
    \toprule
    \multicolumn{3}{c}{Corpus}
      & \multicolumn{2}{c}{Input-query structure (P50/P90)}
      & \multicolumn{7}{c}{Pairs containing feature (\%)} \\
    \cmidrule(lr){1-3}\cmidrule(lr){4-5}\cmidrule(lr){6-12}
    \shortstack[l]{Benchmark\\family}
      & Pairs
      & \shortstack{Schemas\\(tables)}
      & \shortstack{Query\\operators}
      & \shortstack{Join\\occurrences}
      & \shortstack{Agg./\\dist.}
      & \shortstack{Subq.\\(corr.)}
      & \shortstack{Set\\ops}
      & \shortstack{Outer\\join}
      & \shortstack{Order/\\slice}
      & Top-$k$
      & Window \\
    \midrule
    \tableonebenchmarkrows
    \bottomrule
  \end{tabular*}

  \vspace{2pt}
  \begin{minipage}{0.99\textwidth}
    \scriptsize
    \emph{Schemas (tables)} reports distinct schemas and the range of table
    counts.  The \emph{query operators} metric counts projections, scans,
    joins, filters, aggregation or duplicate elimination, set operations,
    ordering/slicing, windows, and \texttt{VALUES}.  Nested queries are
    included, with each CTE definition counted once.  The \emph{join
    occurrences} metric counts explicit joins and comma-separated
    \texttt{FROM} items.  Pair-level P50/P90 use the larger query.  Feature
    percentages count a pair when either query contains the feature; shading
    shows prevalence, and columns may overlap.  \emph{Agg./dist.} covers
    aggregation, grouping, and duplicate elimination; parentheses under
    \emph{Subq.} report correlated subqueries.  Top-$k$ counts
    \texttt{LIMIT}/\texttt{FETCH}, but not \texttt{OFFSET} alone.
  \end{minipage}
\end{table*}

%% file: tables/table1-benchmark-statistics.tex
\newcommand{\tableonebenchmarkrows}{%
Literature~\cite{he2024verieql} & 30 & 24 (1--4) & 6/9 & 1/2 & \benchheat{11}{50\%} & \benchheat{3}{13.3\% (0\%)} & \benchheat{3}{13.3\%} & \benchheat{0}{0\%} & \benchheat{0}{0\%} & \benchheat{0}{0\%} & \benchheat{0}{0\%} \\
Calcite~\cite{begoli2018calcite,he2024verieql} & 236 & 1 (6) & 6/10 & 1/2 & \benchheat{9}{42.4\%} & \benchheat{2}{11\% (2.5\%)} & \benchheat{2}{10.2\%} & \benchheat{6}{26.7\%} & \benchheat{2}{7.6\%} & \benchheat{1}{3\%} & \benchheat{0}{0.8\%} \\
R-Bot TPC-H~\cite{sun2025rbot,tpc2017tpch} & 22 & 1 (8) & 12/18 & 1/5 & \benchheat{22}{100\%} & \benchheat{10}{45.5\% (27.3\%)} & \benchheat{0}{0\%} & \benchheat{1}{4.5\%} & \benchheat{22}{100\%} & \benchheat{22}{100\%} & \benchheat{0}{0\%} \\
R-Bot DSB~\cite{sun2025rbot,ding2021dsb} & 37 & 1 (25) & 22/50 & 6/11 & \benchheat{21}{97.3\%} & \benchheat{8}{35.1\% (21.6\%)} & \benchheat{4}{18.9\%} & \benchheat{2}{10.8\%} & \benchheat{21}{94.6\%} & \benchheat{17}{75.7\%} & \benchheat{0}{0\%} \\
TPC-DS variants~\cite{tpc2024tpcds} & 14 & 1 (25) & 34.5/72 & 5/15 & \benchheat{22}{100\%} & \benchheat{6}{28.6\% (14.3\%)} & \benchheat{20}{92.9\%} & \benchheat{6}{28.6\%} & \benchheat{22}{100\%} & \benchheat{22}{100\%} & \benchheat{8}{35.7\%} \\
WeTune~\cite{wang2022wetune} & 50 & 7 (22--332) & 9.5/17 & 1/3 & \benchheat{10}{44\%} & \benchheat{16}{72\% (26\%)} & \benchheat{6}{26\%} & \benchheat{6}{28\%} & \benchheat{7}{30\%} & \benchheat{4}{20\%} & \benchheat{0}{0\%} \\
}

%% file: sections/implementation.tex
\section{Implementation}
\label{sec:implementation}

\logos\footnote{\url{https://github.com/WindOctober/Logos}} comprises a typed
SQL frontend, the FormalSQL Rocq library, and an
untrusted proof-search loop.  SQLGlot translates the declared input dialect to
Calcite-compatible SQL~\cite{sqlglot2026}; Calcite resolves types and operators
~\cite{begoli2018calcite}; and the Rust frontend lowers Calcite's intermediate
representation to the Rocq representation of the typed core shown in
Figure~\ref{fig:sql-admission-core}.

In Rocq, we refactor and extend SQLCoq~\cite{benzaken2019coqsql} into
FormalSQL, an order-sensitive semantic library.  FormalSQL mechanizes the
complete execution and error semantics underlying
Figure~\ref{fig:core-query-semantics} and provides the bridge between bag and
ordered-list equivalence developed in Section~\ref{sec:bag-abstraction}.  Its
library contains over 2,000 verified lemmas.  They primarily support three
parts of verification: compositional query rewrites, such as merging adjacent
projections; reasoning about scalar and grouped expressions under schema
constraints, such as proving two scalar expressions equivalent in their
respective typed theories; and semantic lifting, such as transporting local
equivalence proofs across query constructors and enclosing contexts while
preserving errors and observable order.

For each input SQL query pair, \logos generates a Rocq verification task that
formalizes its schema constraints, queries, and equivalence statement under
FormalSQL semantics.  It then invokes a proof agent, which determines whether
to prove equivalence or establish non-equivalence.  To focus the agent on proof
construction rather than repository navigation, \logos provides a compact
FormalSQL semantic context together with semantic indexes for the lemma
library.

Non-equivalence cannot in general be certified merely by generating a database
instance and executing it in a deterministic engine such as
PostgreSQL~\cite{postgresql18select}.  FormalSQL may assign a query a set of
possible ordered-list results, whereas one concrete execution observes only a
single result.  \logos therefore accepts non-equivalence only through a
countermodel checked against the complete FormalSQL semantics or an observable
mismatch between the queries' output signatures.

%% file: sections/evaluation.tex
\section{Evaluation}
\label{sec:evaluation}

We evaluate \logos on 389 query pairs from six benchmark families and compare
it with three unbounded SQL verifiers.  Our evaluation addresses two questions:
\begin{itemize}
  \item \textbf{RQ1 (Effectiveness):} How does the verification
  coverage of \logos compare with existing verifiers?
  \item \textbf{RQ2 (Efficiency):} What is the runtime of
  \logos, and where are its main performance bottlenecks?
\end{itemize}

\subsection{Experimental Setup}

\paragraph{Benchmarks.}
Our corpus combines six complementary benchmark families.  The
\emph{Literature} family contains 30 standard-SQL query pairs from
VeriEQL's Literature benchmark~\cite{he2024verieql}, while \emph{Calcite}
contains 236 pairs derived from Apache Calcite optimizer
tests~\cite{begoli2018calcite,he2024verieql}.  The
R-Bot TPC-H and R-Bot DSB families contain 22 and 37 decision-support query
pairs, respectively.  For these families, we generate each target
reproducibly by applying a single rewrite with a pinned version of R-Bot's
Calcite rewrite engine~\cite{sun2025rbot}.  They are therefore deterministic
rewrite workloads rather than unpublished LLM outputs from the R-Bot study.
The \emph{TPC-DS variants} family pairs 14 official base templates with their
corresponding variants.  Finally, \emph{WeTune} contributes 50 before/after
rewrites taken from commits in seven open-source applications~\cite{wang2022wetune}.

\input{tables/table2-cross-benchmark-results}

Table~\ref{tab:benchmark-characteristics} shows that corpus complexity is
multidimensional.  Literature and Calcite account for 266 of the 389 pairs,
but their individual queries are relatively small: their P90 sizes are 9 and
10 operators, respectively, with two join occurrences in both families.
WeTune shifts the stress to application context:
its seven schemas contain 22--332 tables, and 72\% of its pairs contain
subqueries.  R-Bot DSB and the TPC-DS variants instead supply substantially
larger analytical queries, with P90 operator counts of 50 and 72 and join counts
of 11 and 15.  TPC-DS also provides the broadest feature interactions: every
pair combines aggregation, ordering, and top-$k$, while 92.9\% contain set
operations and 35.7\% contain windows.  Thus, the corpus progresses from
localized optimizer rewrites to feature-rich analytical workloads and tests
order sensitivity in composition with other SQL operators.

\paragraph{Environment.}
We run all experiments on a Linux server with an AMD Ryzen 9 5950X CPU
(16 cores and 32 hardware threads) and 128~GB of memory.  Every tool receives
a 14,400-second wall-clock budget and a 16~GB memory limit per query pair,
and we report per-case end-to-end time rather than campaign makespan.  \logos
uses the \texttt{gpt-5.6-sol} proof agent at medium reasoning effort and
permits at most three counterexample rounds.  Each proof-agent sandbox is limited
to 6~GB of memory and 2~GB of writable storage;
proof search and trusted Rocq checking are capped at 14,100 and 420 seconds,
respectively.

\paragraph{Baselines.}
We compare \logos with three unbounded SQL equivalence verifiers using their
original implementations:
\begin{itemize}
  \item \textbf{Cosette} uses Rosette, a solver-aided language embedded in
  Racket~\cite{torlak2014rosette}, for bounded counterexample search, and Coq
  for checked proofs over K-relations and UniNomials~\cite{chu2017cosette}.

  \item \textbf{SQLSolver} reduces U-expressions with unbounded summations to
  an extension of linear integer arithmetic discharged by SMT
  solvers~\cite{ding2023sqlsolver}.

  \item \textbf{QED} normalizes queries into a sum-product normal form and
  proves equality through stabilization and recursive unification, using SMT
  solvers for first-order side conditions~\cite{wang2024qed}.
\end{itemize}

Cosette requires its dedicated DSL, whereas SQLSolver and QED accept
Calcite-compatible SQL subsets.  Since the benchmark families span several
source dialects, we build tool-specific materialization adapters on SQLGlot to
translate each case's declared dialect into the corresponding schema and query
pair.  This lets every verifier run through its intended frontend.

\newcommand{\rqanswer}[2]{%
  \par\addvspace{0.5\baselineskip}
  \begingroup
  \setlength{\fboxrule}{0.8pt}%
  \setlength{\fboxsep}{5pt}%
  \noindent\fcolorbox{rqframe}{rqfill}{%
    \parbox{\dimexpr\linewidth-2\fboxsep-2\fboxrule\relax}{%
      \textbf{Answer to #1.} #2}}%
  \endgroup
  \par\addvspace{0.5\baselineskip}}

\subsection{RQ1: Comparison with Existing Verifiers}

Table~\ref{tab:cross-benchmark-results} reports coverage by benchmark family.
Across the full corpus, \logos solves 338 pairs (\logosrate), compared with 249
(\bestbaselinerate) for SQLSolver, the strongest baseline---an improvement of
\logosimprovement percentage points.  \logos attains the highest solved count
on every family, tying SQLSolver only on R-Bot TPC-H.  The advantage grows on
the 104 order-sensitive pairs identified by Table~\ref{tab:benchmark-characteristics}:
\logos solves 76 (73.1\%), whereas SQLSolver solves 43 (41.3\%), QED 35
(33.7\%), and Cosette 5 (4.8\%).  Thus the aggregate improvement is not obtained
solely from the predominantly unordered Calcite family.

Examining the solved-case intersections, \logos solves 72 cases that no other
evaluated verifier solves, compared with 13 for SQLSolver, four for QED,
and none for Cosette.  The baseline portfolio collectively covers another six
cases through the overlap of SQLSolver and QED, so its total contribution beyond
\logos is 23 cases.

Verification coverage also depends on whether a tool can carry a case through
its input and semantic interfaces.  Table~\ref{tab:cross-benchmark-results}
uses \emph{Unsupported} for non-definite outcomes other than timeouts, including
syntax and lowering failures, model or prover limitations, and solver errors.
Only four \logos cases (1.0\%) fall into this category, compared with 118
(30.3\%) for SQLSolver, 176 (45.2\%) for QED, and 339 (87.1\%) for Cosette.
Thus, \logos carries 99.0\% of the corpus either to a definite result or into
proof search, even when that search eventually exhausts its budget.

The cross-tool intersections show that this broader support translates into
verified results.  Of the 118 pairs unsupported by SQLSolver, \logos solves 94
(79.7\%); another 20 reach proof search but time out.  Conversely, among the
108 pairs solved by \logos but not by SQLSolver, 94 (87.0\%) lie outside
SQLSolver's supported fragment, while only 14 are SQLSolver timeouts.  The
principal gain is therefore semantic reach rather than merely a stronger
decision procedure over the same fragment: a comparatively comprehensive
typed core, combined with agent-guided proof construction, makes
proof-producing verification effective on application-scale query pairs that
existing models cannot represent.

This support gap is most visible in the 101 feature-rich cases from R-Bot DSB,
TPC-DS variants, and WeTune: \logos solves 67, compared with 36 for QED, 27 for
SQLSolver, and none for Cosette.  QED is designed to establish equivalence rather than certify
inequivalence, and its bag-semantics model leaves ordering and concrete
aggregate behavior largely abstract.  SQLSolver's linear-arithmetic reduction
loses completeness on nonlinear expressions, supports only a restricted range
of aggregate reasoning, and handles ordering when the two queries have
compatible outer ordering and slicing structures rather than compositionally
nested top-$k$.
Cosette is further constrained by its DSL and integer-encoded,
unordered relational model.  By contrast, \logos covers these features within
one typed, order-sensitive semantics.  Its four unsupported cases arise from
isolated frontend and formalization gaps; we analyze the much larger set of
timeouts, together with the resulting optimization opportunities, in RQ2.

\begin{figure}[!t]
  \centering
  \includegraphics[width=\linewidth]{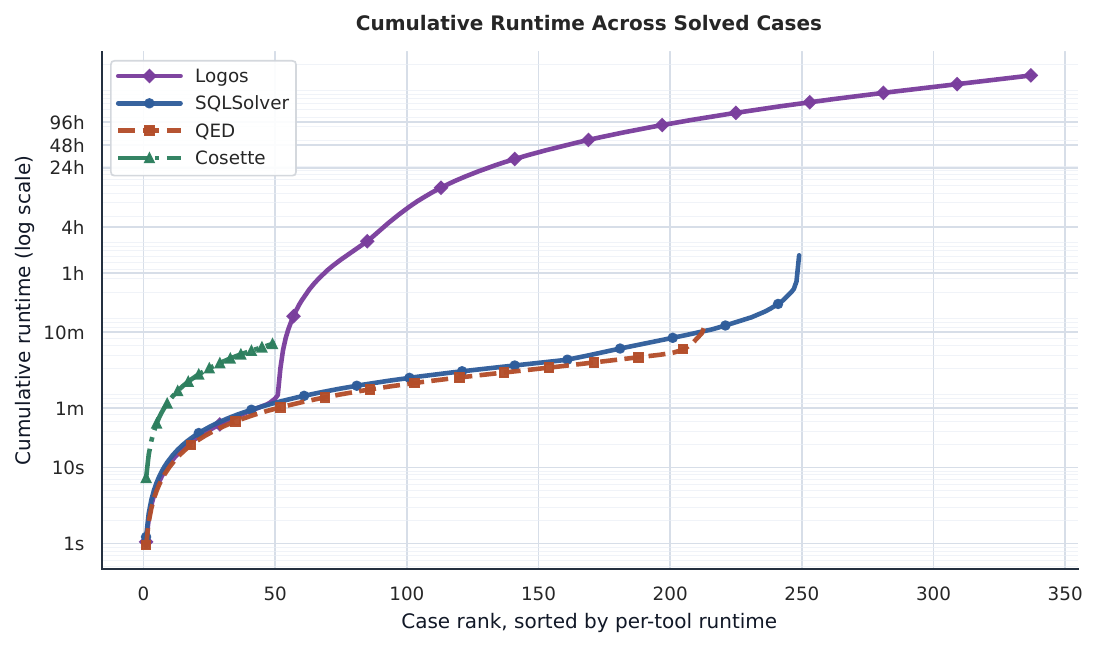}
  \caption{Cumulative runtime on solved cases, sorted by each tool's
  per-case wall-clock time.}
  \Description{A cumulative runtime plot comparing Logos, SQLSolver, QED, and
  Cosette. Logos solves the most cases but accumulates substantially more time.}
  \label{fig:cumulative-runtime}
\end{figure}

\rqanswer{RQ1}{\logos provides the greatest coverage overall and in every
benchmark family, solving \logosimprovement percentage points more pairs than
the strongest unbounded baseline while leaving only 1.0\% of the corpus
unsupported. Its advantage is particularly pronounced for order-sensitive
pairs, where it improves the solved rate from 41.3\% to 73.1\%.}

\subsection{RQ2: Runtime and Bottlenecks}

Figure~\ref{fig:cumulative-runtime} compares end-to-end wall-clock time only on
the cases each tool solves.
The baseline medians are 1.4 seconds for QED, 1.7 seconds for SQLSolver, and
8.5 seconds for Cosette.  \logos is substantially slower: its median is
3,632 seconds (60.5 minutes), and its 338 solved cases accumulate 396.3 hours,
compared with 1.70 hours for SQLSolver, 0.20 hours for QED, and 0.12 hours for
Cosette.  This gap is the tradeoff of \logos{}'s Rocq-plus-agent route: the same
general semantics that yields high solved and support rates also leaves proof
construction to a comparatively expensive synthesis loop.

We analyze all 332 cases that enter agent-guided verification, including both
completed proofs and timeouts.  Their proof-agent rounds consume 563.3 hours.
Of this time, 493.2 hours (87.6\%) occur outside diagnostic Rocq invocations and
cover model deliberation and proof orchestration; 70.1 hours (12.4\%) are spent
compiling candidate proofs.  The bottleneck is therefore proof synthesis rather
than Rocq compilation alone.

The 338 solved cases comprise 243 equivalence results and 95 non-equivalence
results.  Every equivalence result has a Rocq certificate; among these 243
certificate modules, 224 (92.2\%) directly invoke the \logos-specific lemma
library, while 58 (23.9\%) use a bag-to-list bridge.

Each diagnostic request submits the current candidate Rocq module to the
checker.  Across all 332 agent cases, 53,565 such submissions contain
20,834,738 lines of candidate Rocq source in total, with a mean of 62,755 lines
and a median of 32,002 lines per case.  Of these submissions, 50,841 do not
compile, consuming
67.9 of the 70.1 diagnostic hours.  The agent thus automates Rocq proof
development, but currently pays for that automation through many iterations
over incomplete proof states.  The lemma library and bag-to-list bridges
reduce this repeated work by packaging proof-term combinations that would
otherwise have to be reconstructed from the execution semantics.

This cost is not an intended property of the semantics.  Off-the-shelf SMT
theories do not natively provide the combination required here: a set of
alternative legal row sequences together with the type- and NULL-aware
semantics of SQL aggregates such as AVG, MIN, and MAX.  cvc5 provides distinct
finite-bag and sequence theories~\cite{barbosa2022cvc5}, but their primitives
define neither the correspondence between a bag and all row orders admitted by
SQL nor SQL-specific aggregate behavior.  An SMT encoding could represent
result sequences using arrays and explicit lengths, but equivalence proofs
would then require auxiliary invariants relating indices, legal orders, and
aggregate states, together with solver guidance to instantiate these
invariants across operators.  Rocq lets us define this combination directly
and kernel-check the result, while the measurements identify deterministic
proof templates, more selective lemma retrieval, and incremental diagnostics
as the main avenues for reducing synthesis cost.

\rqanswer{RQ2}{The existing unbounded verifiers are much faster on the cases
they solve. \logos trades that speed for broader, proof-producing coverage; its
principal bottleneck is LLM-guided proof construction: 87.6\% of agent-round
time lies outside Rocq diagnostic compilation.  Reusable lemmas and bag-to-list
bridges already reduce repeated reasoning, while deterministic fast paths,
better retrieval, and incremental diagnostics provide the main opportunities
for narrowing the remaining gap.}

\subsection{Threats to Validity}
\label{sec:threats-to-validity}

Our benchmark families span several SQL dialects.  Although the tool-specific
materialization adapters use SQLGlot to normalize each case, the transpiler
cannot guarantee semantic preservation across dialects.  A small number of
cases may therefore acquire semantic differences during translation, leading
to false positive or false negative benchmark outcomes.  We mitigate this risk
by designing \logos around a broad common SQL language and, where the SQL
standard leaves behavior unspecified, adopting PostgreSQL's documented
execution behavior.  This choice improves consistency across the corpus but
does not eliminate dialect-specific discrepancies.

Agent-guided proof search is also not fully deterministic.  Variation in model
serving, response latency, and proof-step selection can change both the route
taken through the proof space and the wall-clock time of an individual run.
Consequently, the reported runtime distribution and timeout boundary may vary
under repeated campaigns, even with fixed prompts and budgets.

%% file: tables/table2-cross-benchmark-results.tex
\begin{table*}[!t]
  \centering
  \caption{Comparison of automated SQL equivalence verifiers by benchmark family.}
  \label{tab:cross-benchmark-results}
  \scriptsize
  \setlength{\tabcolsep}{1.8pt}
  \renewcommand{\arraystretch}{1.14}
  \begin{tabular*}{\textwidth}{@{\extracolsep{\fill}}lr
      ccc ccc ccc ccc@{}}
    \toprule
    \resultHeader{Benchmark family} &
    \resultHeader{Pairs} &
    \multicolumn{3}{c}{\textbf{\logos}} &
    \multicolumn{3}{c}{SQLSolver} &
    \multicolumn{3}{c}{QED} &
    \multicolumn{3}{c@{}}{Cosette} \\
    \cmidrule(lr){3-5}\cmidrule(lr){6-8}
    \cmidrule(lr){9-11}\cmidrule(l){12-14}
      & & Solved & Timeout & Unsupported
      & Solved & Timeout & Unsupported
      & Solved & Timeout & Unsupported
      & Solved & Timeout & Unsupported \\
    \midrule
    Literature       & 30
      & \textbf{\resultCount{28}{93.3}} & \resultCount{2}{6.7} & \resultCount{0}{0.0}
      & \resultCount{23}{76.7} & \resultCount{0}{0.0} & \resultCount{7}{23.3}
      & \resultCount{16}{53.3} & \resultCount{0}{0.0} & \resultCount{14}{46.7}
      & \resultCount{14}{46.7} & \resultCount{0}{0.0} & \resultCount{16}{53.3} \\
    Calcite          & 236
      & \textbf{\resultCount{227}{96.2}} & \resultCount{9}{3.8} & \resultCount{0}{0.0}
      & \resultCount{183}{77.5} & \resultCount{1}{0.4} & \resultCount{52}{22.0}
      & \resultCount{156}{66.1} & \resultCount{0}{0.0} & \resultCount{80}{33.9}
      & \resultCount{36}{15.3} & \resultCount{0}{0.0} & \resultCount{200}{84.7} \\
    R-Bot TPC-H      & 22
      & \textbf{\resultCount{16}{72.7}} & \resultCount{6}{27.3} & \resultCount{0}{0.0}
      & \textbf{\resultCount{16}{72.7}} & \resultCount{0}{0.0} & \resultCount{6}{27.3}
      & \resultCount{5}{22.7} & \resultCount{0}{0.0} & \resultCount{17}{77.3}
      & \resultCount{0}{0.0} & \resultCount{0}{0.0} & \resultCount{22}{100.0} \\
    R-Bot DSB        & 37
      & \textbf{\resultCount{24}{64.9}} & \resultCount{12}{32.4} & \resultCount{1}{2.7}
      & \resultCount{9}{24.3} & \resultCount{14}{37.8} & \resultCount{14}{37.8}
      & \resultCount{19}{51.4} & \resultCount{0}{0.0} & \resultCount{18}{48.6}
      & \resultCount{0}{0.0} & \resultCount{0}{0.0} & \resultCount{37}{100.0} \\
    TPC-DS variants  & 14
      & \textbf{\resultCount{9}{64.3}} & \resultCount{5}{35.7} & \resultCount{0}{0.0}
      & \resultCount{3}{21.4} & \resultCount{4}{28.6} & \resultCount{7}{50.0}
      & \resultCount{2}{14.3} & \resultCount{0}{0.0} & \resultCount{12}{85.7}
      & \resultCount{0}{0.0} & \resultCount{0}{0.0} & \resultCount{14}{100.0} \\
    WeTune           & 50
      & \textbf{\resultCount{34}{68.0}} & \resultCount{13}{26.0} & \resultCount{3}{6.0}
      & \resultCount{15}{30.0} & \resultCount{3}{6.0} & \resultCount{32}{64.0}
      & \resultCount{15}{30.0} & \resultCount{0}{0.0} & \resultCount{35}{70.0}
      & \resultCount{0}{0.0} & \resultCount{0}{0.0} & \resultCount{50}{100.0} \\
    \midrule
    Aggregate        & 389
      & \textbf{\resultCount{338}{86.9}} & \resultCount{47}{12.1} & \resultCount{4}{1.0}
      & \resultCount{249}{64.0} & \resultCount{22}{5.7} & \resultCount{118}{30.3}
      & \resultCount{213}{54.8} & \resultCount{0}{0.0} & \resultCount{176}{45.2}
      & \resultCount{50}{12.9} & \resultCount{0}{0.0} & \resultCount{339}{87.1} \\
    \bottomrule
  \end{tabular*}
\end{table*}

%% file: sections/related-work.tex
\section{Related Work}
\label{sec:related-work}

This section reviews order-aware query semantics and SQL equivalence
verification.  Table~\ref{tab:related-work-support} compares representative
foundations and systems along five dimensions.
\emph{Unbounded checking} ranges over databases of arbitrary finite cardinality.
\emph{Alternative ordered lists} records whether a model retains every order
permitted by SQL, while \emph{compositional top-$k$} further requires these
alternatives to survive slicing and remain visible to enclosing queries.  Thus,
although $[a,b]$ and $[b,a]$ denote the same bag, $\mathsf{fetch}_1$ may yield
either $[a]$ or $[b]$.  \emph{Bag semantics} denotes a multiplicity-based query
meaning.  \emph{Typed aggregate semantics} distinguishes concrete, NULL- and
type-aware aggregation from restricted encodings or uninterpreted aggregate
operators.

\begin{table*}[!t]
  \caption{Comparison of order-aware query models and SQL equivalence systems.}
  \label{tab:related-work-support}
  \centering
  \footnotesize
  \setlength{\tabcolsep}{4pt}
  \begin{tabular*}{\textwidth}{@{\extracolsep{\fill}}lccccc@{}}
    \toprule
    System or framework
           & Unbounded checking
           & Alternative ordered lists
           & Compositional top-$k$
           & Bag semantics
           & Typed aggregate semantics \\
    \midrule
    \multicolumn{6}{@{}l}{\emph{Order-aware semantics and optimization}} \\
    Coburn and Weddell~\cite{coburn1993logic}
      & \textemdash & \supportfull & \textemdash
      & \supportnone & \supportnone \\
    Slivinskas et al.~\cite{slivinskas2001foundation,slivinskas2001volcano}
      & \textemdash & \supportpartial & \textemdash
      & \supportfull & \supportnone \\
    Chinaei~\cite{chinaei2007ordered}
      & \textemdash & \supportpartial & \supportpartial
      & \supportfull & \supportpartial \\
    Amarilli et al.~\cite{amarilli2017possible}
      & \textemdash & \supportfull & \supportpartial
      & \supportfull & \supportnone \\
    \addlinespace[2pt]
    \multicolumn{6}{@{}l}{\emph{SQL query-pair verification}} \\
    Cosette/HoTTSQL~\cite{chu2017cosette,chu2017hottsql}
      & \supportfull & \supportnone & \supportnone
      & \supportfull & \supportnone \\
    SQLSolver~\cite{ding2023sqlsolver}
      & \supportfull & \supportpartial & \supportpartial
      & \supportfull & \supportpartial \\
    QED~\cite{wang2024qed}
      & \supportfull & \supportnone & \supportnone
      & \supportfull & \supportnone \\
    VeriEQL~\cite{he2024verieql}
      & \supportnone & \supportpartial & \supportnone
      & \supportfull & \supportpartial \\
    \textbf{\logos}
      & \supportfull & \supportfull & \supportfull
      & \supportfull & \supportfull \\
    \bottomrule
  \end{tabular*}

  \vspace{2pt}
  \begin{minipage}{0.98\textwidth}
    \scriptsize
    \supportfull~supported; \supportpartial~partially supported;
    \supportnone~unsupported; \textemdash~outside scope.  Ratings concern each
    published fragment; \logos's ratings concern its stated typed SQL fragment,
    rather than all PostgreSQL or ISO SQL.
  \end{minipage}
\end{table*}

\paragraph{Ordered relational semantics and order-aware optimization.}
Ordered relational models long predate modern SQL equivalence verifiers.
Coburn and Weddell propose an ordered-list model for graph-based data in which
a query denotes a nonempty set of tuple lists.  Declarative operators can thus
admit alternative legal orders, while lower-level access plans select more
specific results~\cite{coburn1993logic}.  Their objective is
wide-spectrum plan refinement: an implementation is valid when its possible
results are contained in those of the request.

Slivinskas, Jensen, and Snodgrass give each expression a list-valued semantics
and classify plan equivalence as list, multiset, or set equality.  The required
equivalence propagates inward from the query root.  Without an outermost
\texttt{ORDER BY}, the root requires multiset equality; a sort clears
\emph{OrderRequired} for its child, permitting multiset-preserving rules below
it.  For \texttt{ORDER BY} $A$, two results need only have the same multiset and
the same sequence of projected $A$-values, thereby identifying permutations
within ties~\cite{slivinskas2001foundation}.  This treatment does not extend
directly to top-$k$ operators such as \texttt{FETCH}: a prefix can distinguish
tied permutations, requiring either full list equality or an explicit
semantics of alternative lists rather than multiset equality.

Chinaei reduces an ordered-bag SQL algebra to set algebra by adding interpreted
multiplicity and order-priority columns~\cite{chinaei2007ordered}.  This enables
reuse of set-based containment techniques.  An augmented relation, however,
records one computed priority assignment, encoding order within ordinary set
rows.  Similarly, Amarilli et al. represent order-incomplete data as partially
ordered relations whose linear extensions are possible output lists, but each
po-relation fixes one underlying bag~\cite{amarilli2017possible}.  It therefore
cannot represent grouping or tie-sensitive top-$k$ when possible worlds induce
different result bags.

Mechanized bag semantics over lists imposes a known proof cost.
Bag equality becomes equality modulo permutation, while set equality also
requires duplicate elimination.  Proofs therefore proceed by list induction
while carrying permutation and duplicate-elimination invariants through
operators; HoTTSQL reports that this bookkeeping makes even simple rewrite
proofs lengthy and scales poorly to optimizer rules~\cite{chu2017hottsql}.
\logos instead retains all legal ordered lists, including alternatives with
different bags, while recovering multiplicity-level reasoning through a
closure-based lifting from bag equivalence to ordered-list equivalence at bag-closed
boundaries.

\paragraph{SQL query equivalence, rewrite verification, and optimizer testing.}
SQL query equivalence has been studied through proof assistants, algebraic
decision procedures, and solver-backed verification. Recent industrial evidence
nevertheless shows that real-world rewrite verification remains
challenging~\cite{narasayya2026qoverify}.

One line of work uses mechanized semantics and proof-oriented SQL provers.
Cosette and HoTTSQL are closely related systems from the same line of work:
Cosette emphasizes an automated prover pipeline that combines bounded
counterexample search with Coq-checked reasoning over K-relations and
UniNomials for bag-semantics rewrites, while HoTTSQL develops the HoTT-based
denotational semantics in more depth
~\cite{chu2017cosette,chu2017hottsql,chu2018axiomatic}. Their
reasoning is primarily about relational structure and multiplicities; scalar
predicates and expressions are largely treated as uninterpreted functions. The
semantics also does not give compositional meaning to alternative result
orders, leaving order-sensitive operators outside the main model. In addition,
the systems do not model the full behavior of industrial SQL dialects,
including null-sensitive three-valued logic and type coercions.

A second line develops algebraic or solver-backed decision procedures. QED
combines semiring normalization with first-order reasoning and gives unbounded
guarantees for its supported bag-semantics fragment.  It represents operators
such as \texttt{LIMIT} abstractly and can pass through matching instances by
recursive congruence, but does not model their order-sensitive behavior; casts
remain abstract, and aggregate functions are treated as uninterpreted
higher-order operators~\cite{wang2024qed}.  SQLSolver reduces U-expressions to
linear integer arithmetic and adds a divide-and-conquer ordered-bag procedure.
It gives concrete encodings for core aggregates, but its value reasoning is
restricted to the supported linear-integer-arithmetic fragment.  It simplifies ordering
operators, matches two
\texttt{ORDER BY} nodes only when their keys and following
\texttt{LIMIT}/\texttt{OFFSET} structures agree, and recursively invokes its
bag verifier below them~\cite{ding2023sqlsolver}.  Consequently, its sound
ordered reasoning is tied to compatible query syntax rather than a semantics
for arbitrarily nested order-sensitive queries.  VeriEQL instead gives bounded
SMT verification over finite tuple lists and uses bag equivalence for unordered
outputs.  Its five core aggregates have concrete encodings over an
integer/Boolean/NULL value model.  It encodes \texttt{LIMIT}, \texttt{OFFSET}, and \texttt{FETCH},
including in combination with \texttt{ORDER BY}, but observes list order only
for an outermost \texttt{ORDER BY}.  It therefore does not provide a
compositional semantics of legal result orders in which every tie-induced
top-$k$ result remains visible to enclosing queries~\cite{he2024verieql}.

Earlier symbolic systems include SPES,
which targets bag semantics, and EQUITAS, which targets set semantics; both use
SMT-backed symbolic representations to reason about concrete predicates, but
neither provides semantics for order-sensitive
queries~\cite{zhou2022spes,zhou2019equitas}.
The U-semiring decision procedure (UDP) supplies unbounded bag reasoning, while
abstracting scalar and aggregate semantics and omitting order-sensitive
operators~\cite{chu2018axiomatic}.
Testing and fuzzing tools such as NoREC expose optimizer bugs by constructing
differential or metamorphic test oracles~\cite{rigger2020norec}. Such techniques
are effective at finding concrete counterexamples, but they do not provide
universally quantified rewrite proofs.

\paragraph{LLM-assisted theorem proving.}
Recent theorem-proving agents use language models to retrieve relevant facts,
predict proof steps, or coordinate automated provers. LeanDojo studies
retrieval-based theorem proving over Lean libraries~\cite{yang2023leandojo},
and Thor integrates language models with automated theorem provers for formal
proof search~\cite{jiang2022thor}. These advances are orthogonal to our
SQL-specific semantic development, but inform \logos's use of language models
to guide proof search while leaving certificate acceptance to Rocq.

%% file: sections/conclusion.tex
\section{Conclusion}
\label{sec:conclusion}

We presented \logos, a proof-producing verifier for SQL rewrites over arbitrary
finite databases.  FormalSQL gives a mechanized, compositional semantics for
possible ordered lists and observable failures, while its closure-based
bag-to-list lifting theorem enables local multiplicity reasoning without
erasing order.  \logos combines typed lowering with LLM-guided Rocq proof
construction.  Across 389 query pairs, it solves 338 (\logosrate), compared
with 249 (\bestbaselinerate) for the strongest baseline.  These results
demonstrate the practical value of combining ordered-list semantics with local
bag reasoning.

%% file: main.bbl

\begin{thebibliography}{30}


\ifx \showCODEN    \undefined \def \showCODEN     #1{\unskip}     \fi
\ifx \showDOI      \undefined \def \showDOI       #1{#1}\fi
\ifx \showISBNx    \undefined \def \showISBNx     #1{\unskip}     \fi
\ifx \showISBNxiii \undefined \def \showISBNxiii  #1{\unskip}     \fi
\ifx \showISSN     \undefined \def \showISSN      #1{\unskip}     \fi
\ifx \showLCCN     \undefined \def \showLCCN      #1{\unskip}     \fi
\ifx \shownote     \undefined \def \shownote      #1{#1}          \fi
\ifx \showarticletitle \undefined \def \showarticletitle #1{#1}   \fi
\ifx \showURL      \undefined \def \showURL       {\relax}        \fi
\providecommand\bibfield[2]{#2}
\providecommand\bibinfo[2]{#2}
\providecommand\natexlab[1]{#1}
\providecommand\showeprint[2][]{arXiv:#2}

\bibitem[\protect\citeauthoryear{Amarilli, Ba, Deutch, and Senellart}{Amarilli
  et~al\mbox{.}}{2017}]%
        {amarilli2017possible}
\bibfield{author}{\bibinfo{person}{Antoine Amarilli},
  \bibinfo{person}{Mouhamadou~Lamine Ba}, \bibinfo{person}{Daniel Deutch},
  {and} \bibinfo{person}{Pierre Senellart}.} \bibinfo{year}{2017}\natexlab{}.
\newblock \showarticletitle{Possible and Certain Answers for Queries over
  Order-Incomplete Data}. In \bibinfo{booktitle}{\emph{24th International
  Symposium on Temporal Representation and Reasoning}}
  \emph{(\bibinfo{series}{Leibniz International Proceedings in Informatics})},
  Vol.~\bibinfo{volume}{90}. \bibinfo{publisher}{Schloss
  Dagstuhl--Leibniz-Zentrum f{\"{u}}r Informatik}, \bibinfo{pages}{4:1--4:19}.
\newblock
\urldef\tempurl%
\url{https://doi.org/10.4230/LIPIcs.TIME.2017.4}
\showDOI{\tempurl}


\bibitem[\protect\citeauthoryear{Bai, Alsudais, and Li}{Bai
  et~al\mbox{.}}{2023}]%
        {bai2023querybooster}
\bibfield{author}{\bibinfo{person}{Qiushi Bai}, \bibinfo{person}{Sadeem
  Alsudais}, {and} \bibinfo{person}{Chen Li}.} \bibinfo{year}{2023}\natexlab{}.
\newblock \showarticletitle{QueryBooster: Improving SQL Performance Using
  Middleware Services for Human-Centered Query Rewriting}.
\newblock \bibinfo{journal}{\emph{Proceedings of the VLDB Endowment}}
  \bibinfo{volume}{16}, \bibinfo{number}{11} (\bibinfo{year}{2023}),
  \bibinfo{pages}{2911--2924}.
\newblock
\urldef\tempurl%
\url{https://doi.org/10.14778/3611479.3611497}
\showDOI{\tempurl}


\bibitem[\protect\citeauthoryear{Barbosa, Barrett, Brain, Kremer, Lachnitt,
  Mann, Mohamed, Mohamed, Niemetz, N{\"{o}}tzli, Ozdemir, Preiner, Reynolds,
  Sheng, Tinelli, and Zohar}{Barbosa et~al\mbox{.}}{2022}]%
        {barbosa2022cvc5}
\bibfield{author}{\bibinfo{person}{Haniel Barbosa}, \bibinfo{person}{Clark~W.
  Barrett}, \bibinfo{person}{Martin Brain}, \bibinfo{person}{Gereon Kremer},
  \bibinfo{person}{Hanna Lachnitt}, \bibinfo{person}{Makai Mann},
  \bibinfo{person}{Abdalrhman Mohamed}, \bibinfo{person}{Mudathir Mohamed},
  \bibinfo{person}{Aina Niemetz}, \bibinfo{person}{Andres N{\"{o}}tzli},
  \bibinfo{person}{Alex Ozdemir}, \bibinfo{person}{Mathias Preiner},
  \bibinfo{person}{Andrew Reynolds}, \bibinfo{person}{Ying Sheng},
  \bibinfo{person}{Cesare Tinelli}, {and} \bibinfo{person}{Yoni Zohar}.}
  \bibinfo{year}{2022}\natexlab{}.
\newblock \showarticletitle{cvc5: A Versatile and Industrial-Strength {SMT}
  Solver}. In \bibinfo{booktitle}{\emph{Tools and Algorithms for the
  Construction and Analysis of Systems}} \emph{(\bibinfo{series}{Lecture Notes
  in Computer Science})}, Vol.~\bibinfo{volume}{13243}.
  \bibinfo{publisher}{Springer}, \bibinfo{pages}{415--442}.
\newblock
\urldef\tempurl%
\url{https://doi.org/10.1007/978-3-030-99524-9_24}
\showDOI{\tempurl}


\bibitem[\protect\citeauthoryear{Begoli, Camacho-Rodr{\'{i}}guez, Hyde, Mior,
  and Lemire}{Begoli et~al\mbox{.}}{2018}]%
        {begoli2018calcite}
\bibfield{author}{\bibinfo{person}{Edmon Begoli}, \bibinfo{person}{Jes{\'{u}}s
  Camacho-Rodr{\'{i}}guez}, \bibinfo{person}{Julian Hyde},
  \bibinfo{person}{Michael~J. Mior}, {and} \bibinfo{person}{Daniel Lemire}.}
  \bibinfo{year}{2018}\natexlab{}.
\newblock \showarticletitle{Apache Calcite: A Foundational Framework for
  Optimized Query Processing Over Heterogeneous Data Sources}. In
  \bibinfo{booktitle}{\emph{Proceedings of the 2018 International Conference on
  Management of Data}}. \bibinfo{publisher}{Association for Computing
  Machinery}, \bibinfo{pages}{221--230}.
\newblock
\urldef\tempurl%
\url{https://doi.org/10.1145/3183713.3190662}
\showDOI{\tempurl}


\bibitem[\protect\citeauthoryear{Benzaken and Contejean}{Benzaken and
  Contejean}{2019}]%
        {benzaken2019coqsql}
\bibfield{author}{\bibinfo{person}{V{\'{e}}ronique Benzaken} {and}
  \bibinfo{person}{{\'{E}}velyne Contejean}.} \bibinfo{year}{2019}\natexlab{}.
\newblock \showarticletitle{A Coq Mechanised Formal Semantics for Realistic SQL
  Queries: Formally Reconciling SQL and Bag Relational Algebra}. In
  \bibinfo{booktitle}{\emph{Proceedings of the 8th ACM SIGPLAN International
  Conference on Certified Programs and Proofs}}.
  \bibinfo{publisher}{Association for Computing Machinery},
  \bibinfo{pages}{249--261}.
\newblock
\urldef\tempurl%
\url{https://doi.org/10.1145/3293880.3294107}
\showDOI{\tempurl}


\bibitem[\protect\citeauthoryear{Chinaei}{Chinaei}{2007}]%
        {chinaei2007ordered}
\bibfield{author}{\bibinfo{person}{Hamid~R. Chinaei}.}
  \bibinfo{year}{2007}\natexlab{}.
\newblock \emph{\bibinfo{title}{An Ordered Bag Semantics for SQL}}.
\newblock Master's thesis. \bibinfo{school}{University of Waterloo}.
\newblock
\urldef\tempurl%
\url{https://hdl.handle.net/10012/3062}
\showURL{%
\tempurl}


\bibitem[\protect\citeauthoryear{Chu, Murphy, Roesch, Cheung, and Suciu}{Chu
  et~al\mbox{.}}{2018}]%
        {chu2018axiomatic}
\bibfield{author}{\bibinfo{person}{Shumo Chu}, \bibinfo{person}{Brendan
  Murphy}, \bibinfo{person}{Jared Roesch}, \bibinfo{person}{Alvin Cheung},
  {and} \bibinfo{person}{Dan Suciu}.} \bibinfo{year}{2018}\natexlab{}.
\newblock \showarticletitle{Axiomatic Foundations and Algorithms for Deciding
  Semantic Equivalences of SQL Queries}.
\newblock \bibinfo{journal}{\emph{Proceedings of the VLDB Endowment}}
  \bibinfo{volume}{11}, \bibinfo{number}{11} (\bibinfo{year}{2018}),
  \bibinfo{pages}{1482--1495}.
\newblock
\urldef\tempurl%
\url{https://doi.org/10.14778/3236187.3236200}
\showDOI{\tempurl}


\bibitem[\protect\citeauthoryear{Chu, Wang, Weitz, and Cheung}{Chu
  et~al\mbox{.}}{2017a}]%
        {chu2017cosette}
\bibfield{author}{\bibinfo{person}{Shumo Chu}, \bibinfo{person}{Chenglong
  Wang}, \bibinfo{person}{Konstantin Weitz}, {and} \bibinfo{person}{Alvin
  Cheung}.} \bibinfo{year}{2017}\natexlab{a}.
\newblock \showarticletitle{Cosette: An Automated Prover for SQL}. In
  \bibinfo{booktitle}{\emph{8th Biennial Conference on Innovative Data Systems
  Research}}. \bibinfo{publisher}{www.cidrdb.org}.
\newblock
\urldef\tempurl%
\url{https://www.cidrdb.org/cidr2017/papers/p51-chu-cidr17.pdf}
\showURL{%
\tempurl}


\bibitem[\protect\citeauthoryear{Chu, Weitz, Cheung, and Suciu}{Chu
  et~al\mbox{.}}{2017b}]%
        {chu2017hottsql}
\bibfield{author}{\bibinfo{person}{Shumo Chu}, \bibinfo{person}{Konstantin
  Weitz}, \bibinfo{person}{Alvin Cheung}, {and} \bibinfo{person}{Dan Suciu}.}
  \bibinfo{year}{2017}\natexlab{b}.
\newblock \showarticletitle{HoTTSQL: Proving Query Rewrites with Univalent SQL
  Semantics}. In \bibinfo{booktitle}{\emph{Proceedings of the 38th ACM SIGPLAN
  Conference on Programming Language Design and Implementation}}.
  \bibinfo{publisher}{Association for Computing Machinery},
  \bibinfo{pages}{510--524}.
\newblock
\urldef\tempurl%
\url{https://doi.org/10.1145/3062341.3062348}
\showDOI{\tempurl}


\bibitem[\protect\citeauthoryear{Coburn and Weddell}{Coburn and
  Weddell}{1993}]%
        {coburn1993logic}
\bibfield{author}{\bibinfo{person}{Neil Coburn} {and} \bibinfo{person}{Grant~E.
  Weddell}.} \bibinfo{year}{1993}\natexlab{}.
\newblock \showarticletitle{A Logic for Rule-Based Query Optimization in
  Graph-Based Data Models}. In \bibinfo{booktitle}{\emph{Deductive and
  Object-Oriented Databases}} \emph{(\bibinfo{series}{Lecture Notes in Computer
  Science})}, Vol.~\bibinfo{volume}{760}. \bibinfo{publisher}{Springer},
  \bibinfo{pages}{120--145}.
\newblock
\urldef\tempurl%
\url{https://doi.org/10.1007/3-540-57530-8_8}
\showDOI{\tempurl}


\bibitem[\protect\citeauthoryear{Ding, Chaudhuri, Gehrke, and Narasayya}{Ding
  et~al\mbox{.}}{2021}]%
        {ding2021dsb}
\bibfield{author}{\bibinfo{person}{Bailu Ding}, \bibinfo{person}{Surajit
  Chaudhuri}, \bibinfo{person}{Johannes Gehrke}, {and} \bibinfo{person}{Vivek
  Narasayya}.} \bibinfo{year}{2021}\natexlab{}.
\newblock \showarticletitle{{DSB}: A Decision Support Benchmark for
  Workload-Driven and Traditional Database Systems}.
\newblock \bibinfo{journal}{\emph{Proceedings of the VLDB Endowment}}
  \bibinfo{volume}{14}, \bibinfo{number}{13} (\bibinfo{year}{2021}),
  \bibinfo{pages}{3376--3388}.
\newblock
\urldef\tempurl%
\url{https://doi.org/10.14778/3484224.3484234}
\showDOI{\tempurl}


\bibitem[\protect\citeauthoryear{Ding, Wang, Yang, Zhang, Xu, Chen, Piskac, and
  Li}{Ding et~al\mbox{.}}{2023}]%
        {ding2023sqlsolver}
\bibfield{author}{\bibinfo{person}{Haoran Ding}, \bibinfo{person}{Zhaoguo
  Wang}, \bibinfo{person}{Yicun Yang}, \bibinfo{person}{Dexin Zhang},
  \bibinfo{person}{Zhenglin Xu}, \bibinfo{person}{Haibo Chen},
  \bibinfo{person}{Ruzica Piskac}, {and} \bibinfo{person}{Jinyang Li}.}
  \bibinfo{year}{2023}\natexlab{}.
\newblock \showarticletitle{Proving Query Equivalence Using Linear Integer
  Arithmetic}.
\newblock \bibinfo{journal}{\emph{Proceedings of the ACM on Management of
  Data}} \bibinfo{volume}{1}, \bibinfo{number}{4}, Article
  \bibinfo{articleno}{227} (\bibinfo{year}{2023}),
  \bibinfo{numpages}{26}~pages.
\newblock
\urldef\tempurl%
\url{https://doi.org/10.1145/3626768}
\showDOI{\tempurl}


\bibitem[\protect\citeauthoryear{He, Zhao, Wang, and Wang}{He
  et~al\mbox{.}}{2024}]%
        {he2024verieql}
\bibfield{author}{\bibinfo{person}{Yang He}, \bibinfo{person}{Pinhan Zhao},
  \bibinfo{person}{Xinyu Wang}, {and} \bibinfo{person}{Yuepeng Wang}.}
  \bibinfo{year}{2024}\natexlab{}.
\newblock \showarticletitle{VeriEQL: Bounded Equivalence Verification for
  Complex SQL Queries with Integrity Constraints}.
\newblock \bibinfo{journal}{\emph{Proceedings of the ACM on Programming
  Languages}} \bibinfo{volume}{8}, \bibinfo{number}{OOPSLA1}
  (\bibinfo{year}{2024}), \bibinfo{pages}{1071--1099}.
\newblock
\urldef\tempurl%
\url{https://doi.org/10.1145/3649849}
\showDOI{\tempurl}


\bibitem[\protect\citeauthoryear{Jiang, Li, Tworkowski, Czechowski,
  Odrzyg{\'{o}}{\'{z}}d{\'{z}}, Mi{\l}o{\'{s}}, Wu, and Jamnik}{Jiang
  et~al\mbox{.}}{2022}]%
        {jiang2022thor}
\bibfield{author}{\bibinfo{person}{Albert~Qiaochu Jiang},
  \bibinfo{person}{Wenda Li}, \bibinfo{person}{Szymon Tworkowski},
  \bibinfo{person}{Konrad Czechowski}, \bibinfo{person}{Tomasz
  Odrzyg{\'{o}}{\'{z}}d{\'{z}}}, \bibinfo{person}{Piotr Mi{\l}o{\'{s}}},
  \bibinfo{person}{Yuhuai Wu}, {and} \bibinfo{person}{Mateja Jamnik}.}
  \bibinfo{year}{2022}\natexlab{}.
\newblock \showarticletitle{Thor: Wielding Hammers to Integrate Language Models
  and Automated Theorem Provers}. In \bibinfo{booktitle}{\emph{Advances in
  Neural Information Processing Systems 35}}. \bibinfo{pages}{8360--8373}.
\newblock
\urldef\tempurl%
\url{https://doi.org/10.52202/068431-0608}
\showDOI{\tempurl}


\bibitem[\protect\citeauthoryear{Mao and Contributors}{Mao and
  Contributors}{2026}]%
        {sqlglot2026}
\bibfield{author}{\bibinfo{person}{Toby Mao} {and}
  \bibinfo{person}{Contributors}.} \bibinfo{year}{2026}\natexlab{}.
\newblock \bibinfo{title}{SQLGlot: Python SQL Parser and Transpiler, Version
  30.11.0}.
\newblock
\newblock
\urldef\tempurl%
\url{https://pypi.org/project/sqlglot/30.11.0/}
\showURL{%
\tempurl}
\newblock
\shownote{Accessed: 2026-07-07.}


\bibitem[\protect\citeauthoryear{Narasayya and Chaudhuri}{Narasayya and
  Chaudhuri}{2026}]%
        {narasayya2026qoverify}
\bibfield{author}{\bibinfo{person}{Vivek~R. Narasayya} {and}
  \bibinfo{person}{Surajit Chaudhuri}.} \bibinfo{year}{2026}\natexlab{}.
\newblock \showarticletitle{Leveraging Query Optimizers to Verify the Soundness
  of LLM-Based Query Rewrites for Real-World Workloads, and More!}. In
  \bibinfo{booktitle}{\emph{16th Conference on Innovative Data Systems
  Research}}. \bibinfo{publisher}{www.cidrdb.org}, \bibinfo{address}{Chaminade,
  CA, USA}, 12.
\newblock
\urldef\tempurl%
\url{https://www.cidrdb.org/cidr2026/papers/p33-narasayya.pdf}
\showURL{%
\tempurl}


\bibitem[\protect\citeauthoryear{Pirahesh, Hellerstein, and Hasan}{Pirahesh
  et~al\mbox{.}}{1992}]%
        {pirahesh1992starburst}
\bibfield{author}{\bibinfo{person}{Hamid Pirahesh}, \bibinfo{person}{Joseph~M.
  Hellerstein}, {and} \bibinfo{person}{Waqar Hasan}.}
  \bibinfo{year}{1992}\natexlab{}.
\newblock \showarticletitle{Extensible/Rule Based Query Rewrite Optimization in
  Starburst}. In \bibinfo{booktitle}{\emph{Proceedings of the 1992 ACM SIGMOD
  International Conference on Management of Data}}.
  \bibinfo{publisher}{Association for Computing Machinery},
  \bibinfo{pages}{39--48}.
\newblock
\urldef\tempurl%
\url{https://doi.org/10.1145/130283.130294}
\showDOI{\tempurl}


\bibitem[\protect\citeauthoryear{{PostgreSQL Global Development
  Group}}{{PostgreSQL Global Development Group}}{2026}]%
        {postgresql18select}
\bibfield{author}{\bibinfo{person}{{PostgreSQL Global Development Group}}.}
  \bibinfo{year}{2026}\natexlab{}.
\newblock \bibinfo{booktitle}{\emph{{PostgreSQL} 18.4: {SELECT}}}.
\newblock
\urldef\tempurl%
\url{https://www.postgresql.org/docs/18/sql-select.html}
\showURL{%
\tempurl}
\newblock
\shownote{Accessed: 2026-07-28.}


\bibitem[\protect\citeauthoryear{Rigger and Su}{Rigger and Su}{2020}]%
        {rigger2020norec}
\bibfield{author}{\bibinfo{person}{Manuel Rigger} {and}
  \bibinfo{person}{Zhendong Su}.} \bibinfo{year}{2020}\natexlab{}.
\newblock \showarticletitle{Detecting Optimization Bugs in Database Engines via
  Non-Optimizing Reference Engine Construction}. In
  \bibinfo{booktitle}{\emph{Proceedings of the 28th ACM Joint Meeting on
  European Software Engineering Conference and Symposium on the Foundations of
  Software Engineering}}. \bibinfo{publisher}{Association for Computing
  Machinery}, \bibinfo{pages}{1140--1152}.
\newblock
\urldef\tempurl%
\url{https://doi.org/10.1145/3368089.3409710}
\showDOI{\tempurl}


\bibitem[\protect\citeauthoryear{Slivinskas and Jensen}{Slivinskas and
  Jensen}{2001}]%
        {slivinskas2001volcano}
\bibfield{author}{\bibinfo{person}{Giedrius Slivinskas} {and}
  \bibinfo{person}{Christian~S. Jensen}.} \bibinfo{year}{2001}\natexlab{}.
\newblock \showarticletitle{Enhancing an Extensible Query Optimizer with
  Support for Multiple Equivalence Types}. In
  \bibinfo{booktitle}{\emph{Advances in Databases and Information Systems}}
  \emph{(\bibinfo{series}{Lecture Notes in Computer Science})},
  Vol.~\bibinfo{volume}{2151}. \bibinfo{publisher}{Springer},
  \bibinfo{pages}{55--69}.
\newblock
\urldef\tempurl%
\url{https://doi.org/10.1007/3-540-44803-9_6}
\showDOI{\tempurl}


\bibitem[\protect\citeauthoryear{Slivinskas, Jensen, and Snodgrass}{Slivinskas
  et~al\mbox{.}}{2001}]%
        {slivinskas2001foundation}
\bibfield{author}{\bibinfo{person}{Giedrius Slivinskas},
  \bibinfo{person}{Christian~S. Jensen}, {and} \bibinfo{person}{Richard~T.
  Snodgrass}.} \bibinfo{year}{2001}\natexlab{}.
\newblock \showarticletitle{A Foundation for Conventional and Temporal Query
  Optimization Addressing Duplicates and Ordering}.
\newblock \bibinfo{journal}{\emph{IEEE Transactions on Knowledge and Data
  Engineering}} \bibinfo{volume}{13}, \bibinfo{number}{1}
  (\bibinfo{year}{2001}), \bibinfo{pages}{21--49}.
\newblock
\urldef\tempurl%
\url{https://doi.org/10.1109/69.908979}
\showDOI{\tempurl}


\bibitem[\protect\citeauthoryear{Sun, Zhou, Li, Yu, Feng, and Zhang}{Sun
  et~al\mbox{.}}{2025}]%
        {sun2025rbot}
\bibfield{author}{\bibinfo{person}{Zhaoyan Sun}, \bibinfo{person}{Xuanhe Zhou},
  \bibinfo{person}{Guoliang Li}, \bibinfo{person}{Xiang Yu},
  \bibinfo{person}{Jianhua Feng}, {and} \bibinfo{person}{Yong Zhang}.}
  \bibinfo{year}{2025}\natexlab{}.
\newblock \showarticletitle{R-Bot: An LLM-Based Query Rewrite System}.
\newblock \bibinfo{journal}{\emph{Proceedings of the VLDB Endowment}}
  \bibinfo{volume}{18}, \bibinfo{number}{12} (\bibinfo{year}{2025}),
  \bibinfo{pages}{5031--5044}.
\newblock
\urldef\tempurl%
\url{https://doi.org/10.14778/3750601.3750625}
\showDOI{\tempurl}


\bibitem[\protect\citeauthoryear{Torlak and Bod{\'{\i}}k}{Torlak and
  Bod{\'{\i}}k}{2014}]%
        {torlak2014rosette}
\bibfield{author}{\bibinfo{person}{Emina Torlak} {and}
  \bibinfo{person}{Rastislav Bod{\'{\i}}k}.} \bibinfo{year}{2014}\natexlab{}.
\newblock \showarticletitle{A Lightweight Symbolic Virtual Machine for
  Solver-Aided Host Languages}. In \bibinfo{booktitle}{\emph{Proceedings of the
  35th ACM SIGPLAN Conference on Programming Language Design and
  Implementation}}. \bibinfo{publisher}{Association for Computing Machinery},
  \bibinfo{pages}{530--541}.
\newblock
\urldef\tempurl%
\url{https://doi.org/10.1145/2594291.2594340}
\showDOI{\tempurl}


\bibitem[\protect\citeauthoryear{{Transaction Processing Performance
  Council}}{{Transaction Processing Performance Council}}{2017}]%
        {tpc2017tpch}
\bibfield{author}{\bibinfo{person}{{Transaction Processing Performance
  Council}}.} \bibinfo{year}{2017}\natexlab{}.
\newblock \bibinfo{booktitle}{\emph{{TPC Benchmark H} (Decision Support):
  Standard Specification, Revision 2.17.3}}.
\newblock
\urldef\tempurl%
\url{https://www.tpc.org/tpc_documents_current_versions/pdf/TPC-H_v2.17.3.pdf}
\showURL{%
\tempurl}
\newblock
\shownote{Accessed: 2026-07-31.}


\bibitem[\protect\citeauthoryear{{Transaction Processing Performance
  Council}}{{Transaction Processing Performance Council}}{2024}]%
        {tpc2024tpcds}
\bibfield{author}{\bibinfo{person}{{Transaction Processing Performance
  Council}}.} \bibinfo{year}{2024}\natexlab{}.
\newblock \bibinfo{booktitle}{\emph{{TPC Benchmark DS}: Standard Specification,
  Version 4.0.0}}.
\newblock
\urldef\tempurl%
\url{https://www.tpc.org/TPC_Documents_Current_Versions/pdf/TPC-DS_v4.0.0.pdf}
\showURL{%
\tempurl}
\newblock
\shownote{Accessed: 2026-07-31.}


\bibitem[\protect\citeauthoryear{Wang, Pan, and Cheung}{Wang
  et~al\mbox{.}}{2024}]%
        {wang2024qed}
\bibfield{author}{\bibinfo{person}{Shuxian Wang}, \bibinfo{person}{Sicheng
  Pan}, {and} \bibinfo{person}{Alvin Cheung}.} \bibinfo{year}{2024}\natexlab{}.
\newblock \showarticletitle{QED: A Powerful Query Equivalence Decider for SQL}.
\newblock \bibinfo{journal}{\emph{Proceedings of the VLDB Endowment}}
  \bibinfo{volume}{17}, \bibinfo{number}{11} (\bibinfo{year}{2024}),
  \bibinfo{pages}{3602--3614}.
\newblock
\urldef\tempurl%
\url{https://doi.org/10.14778/3681954.3682024}
\showDOI{\tempurl}


\bibitem[\protect\citeauthoryear{Wang, Zhou, Yang, Ding, Hu, Ding, Tang, Chen,
  and Li}{Wang et~al\mbox{.}}{2022}]%
        {wang2022wetune}
\bibfield{author}{\bibinfo{person}{Zhaoguo Wang}, \bibinfo{person}{Zhou Zhou},
  \bibinfo{person}{Yicun Yang}, \bibinfo{person}{Haoran Ding},
  \bibinfo{person}{Gansen Hu}, \bibinfo{person}{Ding Ding},
  \bibinfo{person}{Chuzhe Tang}, \bibinfo{person}{Haibo Chen}, {and}
  \bibinfo{person}{Jinyang Li}.} \bibinfo{year}{2022}\natexlab{}.
\newblock \showarticletitle{WeTune: Automatic Discovery and Verification of
  Query Rewrite Rules}. In \bibinfo{booktitle}{\emph{Proceedings of the 2022
  International Conference on Management of Data}}.
  \bibinfo{publisher}{Association for Computing Machinery},
  \bibinfo{pages}{94--107}.
\newblock
\urldef\tempurl%
\url{https://doi.org/10.1145/3514221.3526125}
\showDOI{\tempurl}


\bibitem[\protect\citeauthoryear{Yang, Swope, Gu, Chalamala, Song, Yu, Godil,
  Prenger, and Anandkumar}{Yang et~al\mbox{.}}{2023}]%
        {yang2023leandojo}
\bibfield{author}{\bibinfo{person}{Kaiyu Yang}, \bibinfo{person}{Aidan Swope},
  \bibinfo{person}{Alex Gu}, \bibinfo{person}{Rahul Chalamala},
  \bibinfo{person}{Peiyang Song}, \bibinfo{person}{Shixing Yu},
  \bibinfo{person}{Saad Godil}, \bibinfo{person}{Ryan~J. Prenger}, {and}
  \bibinfo{person}{Animashree Anandkumar}.} \bibinfo{year}{2023}\natexlab{}.
\newblock \showarticletitle{LeanDojo: Theorem Proving with Retrieval-Augmented
  Language Models}. In \bibinfo{booktitle}{\emph{Advances in Neural Information
  Processing Systems 36}}. \bibinfo{pages}{21573--21612}.
\newblock
\urldef\tempurl%
\url{https://doi.org/10.52202/075280-0944}
\showDOI{\tempurl}


\bibitem[\protect\citeauthoryear{Zhou, Arulraj, Navathe, Harris, and Wu}{Zhou
  et~al\mbox{.}}{2022}]%
        {zhou2022spes}
\bibfield{author}{\bibinfo{person}{Qi Zhou}, \bibinfo{person}{Joy Arulraj},
  \bibinfo{person}{Shamkant~B. Navathe}, \bibinfo{person}{William Harris},
  {and} \bibinfo{person}{Jinpeng Wu}.} \bibinfo{year}{2022}\natexlab{}.
\newblock \showarticletitle{{SPES}: A Symbolic Approach to Proving Query
  Equivalence Under Bag Semantics}. In \bibinfo{booktitle}{\emph{2022 IEEE 38th
  International Conference on Data Engineering}}. \bibinfo{publisher}{IEEE},
  \bibinfo{pages}{2735--2748}.
\newblock
\urldef\tempurl%
\url{https://doi.org/10.1109/ICDE53745.2022.00250}
\showDOI{\tempurl}


\bibitem[\protect\citeauthoryear{Zhou, Arulraj, Navathe, Harris, and Xu}{Zhou
  et~al\mbox{.}}{2019}]%
        {zhou2019equitas}
\bibfield{author}{\bibinfo{person}{Qi Zhou}, \bibinfo{person}{Joy Arulraj},
  \bibinfo{person}{Shamkant~B. Navathe}, \bibinfo{person}{William Harris},
  {and} \bibinfo{person}{Dong Xu}.} \bibinfo{year}{2019}\natexlab{}.
\newblock \showarticletitle{Automated Verification of Query Equivalence Using
  Satisfiability Modulo Theories}.
\newblock \bibinfo{journal}{\emph{Proceedings of the VLDB Endowment}}
  \bibinfo{volume}{12}, \bibinfo{number}{11} (\bibinfo{year}{2019}),
  \bibinfo{pages}{1276--1288}.
\newblock
\urldef\tempurl%
\url{https://doi.org/10.14778/3342263.3342267}
\showDOI{\tempurl}


\end{thebibliography}
